\documentclass[twocolumn,10pt]{IEEEtran}

\usepackage{amsmath,amssymb,amsthm}
\usepackage{graphicx}
\usepackage{comment}
\usepackage{xcolor}
\usepackage{subfigure}
\usepackage{booktabs}
\usepackage[acronym]{glossaries}
\newacronym{em}{EM}{electromagnetic}
\newacronym{cscg}{CSCG}{circularly-symmetric complex Gaussian}
\newacronym{clt}{CLT}{central limit theorem}
\newacronym{csi}{CSI}{channel state information}
\newacronym{los}{LoS}{line-of-sight}
\newacronym{siso}{SISO}{single-input single-output}
\newacronym{mimo}{MIMO}{multiple-input multiple-output}
\newacronym{snr}{SNR}{signal-to-noise ratio}
\newacronym{awgn}{AWGN}{additive white Gaussian noise}
\newacronym{svd}{SVD}{singular value decomposition}

\newtheorem{lemma}{Lemma}
\newtheorem{proposition}{Proposition}

\begin{document}

%\title{Global Optimal Closed-Form Solutions for\\Beyond Diagonal RIS With Mutual Coupling:\\Is Mutual Coupling Detrimental or Beneficial?}
\title{Global Optimal Closed-Form Solutions for\\Intelligent Surfaces With Mutual Coupling:\\Is Mutual Coupling Detrimental or Beneficial?}

\author{Matteo~Nerini,~\IEEEmembership{Member,~IEEE},
        Hongyu~Li,~\IEEEmembership{Member,~IEEE},
        Bruno~Clerckx,~\IEEEmembership{Fellow,~IEEE}

\thanks{This work has been supported in part by UKRI under Grant EP/Y004086/1, EP/X040569/1, EP/Y037197/1, EP/X04047X/1, EP/Y037243/1. Corresponding author: Bruno Clerckx.}
\thanks{Matteo Nerini and Bruno Clerckx are with the Department of Electrical and Electronic Engineering, Imperial College London, SW7 2AZ London, U.K. (e-mail: m.nerini20@imperial.ac.uk; b.clerckx@imperial.ac.uk).}
\thanks{Bruno Clerckx is also with Kyung Hee University, Seoul, Korea.}
\thanks{Hongyu Li is with the Internet of Things Thrust, The Hong Kong University of Science and Technology (Guangzhou), Guangzhou 511400, China (email: hongyuli@hkust-gz.edu.cn).}}

\maketitle

\begin{abstract}
Reconfigurable Intelligent Surface (RIS) is a breakthrough technology enabling the dynamic control of the propagation environment in wireless communications through programmable surfaces.
To improve the flexibility of conventional diagonal RIS (D-RIS), beyond diagonal RIS (BD-RIS) has emerged as a family of more general RIS architectures.
However, D-RIS and BD-RIS have been commonly explored neglecting mutual coupling effects, while the global optimization of RIS with mutual coupling, its performance limits, and scaling laws remain unexplored.
This study addresses these gaps by deriving global optimal closed-form solutions for BD-RIS with mutual coupling to maximize the channel gain, specifically fully- and tree-connected RISs.
Besides, we provide the expression of the maximum channel gain achievable in the presence of mutual coupling and its scaling law in closed form.
By using the derived scaling laws, we analytically prove that mutual coupling increases the channel gain on average under Rayleigh fading channels.
Our theoretical analysis, confirmed by numerical simulations, shows that both fully- and tree-connected RISs with mutual coupling achieve the same channel gain upper bound when optimized with the proposed global optimal solutions.
Furthermore, we observe that a mutual coupling-unaware optimization of RIS can cause a channel gain degradation of up to 5 dB.
\end{abstract}


%%%%%%%%%%%%%%%%%%%%%%%%%%%%%%%%%%%%%%%%%%%%%%%%%%
\section{Introduction}

% From D-RIS to BD-RIS
Reconfigurable intelligent surface (RIS) is a breakthrough technology that allows dynamic control over the wireless channel in wireless systems \cite{wu21}.
A RIS is a surface made of numerous elements with reconfigurable scattering properties, able to steer the impinging \gls{em} signal toward the intended receiver.
In a conventional RIS architecture, the RIS elements are not interconnected to each other, resulting in RIS being characterized by a diagonal phase shift matrix.
To overcome this limitation of conventional RIS, also denoted as diagonal RIS (D-RIS), novel and more general RIS architectures have emerged under the name of beyond diagonal RIS (BD-RIS) \cite{li23-1}.
The novelty introduced in BD-RIS is the presence of interconnections between the RIS elements, allowing the impinging waves to flow through the surface and the RIS to have a scattering matrix not constrained to be diagonal.

% BD-RIS architectures
Multiple BD-RIS architectures have been developed to offer advanced flexibility and performance over D-RIS \cite{she20,ner23-1}.
Among them, the fully-connected RIS offers the highest flexibility and performance since each RIS element is connected to all others via tunable impedance components \cite{she20}.
To maintain high performance while decreasing the RIS circuit complexity, the tree-connected RIS has been proposed \cite{ner23-1}, which reduces the required number of tunable impedance components.
Optimal BD-RIS architectures that effectively balance performance and circuit complexity have been investigated in \cite{ner23-3}, where the Pareto frontier of this trade-off has been characterized.
Additionally, BD-RIS has been studied not only for performance enhancement but also to enable full-space coverage, leveraging architectures working in hybrid transmissive and reflective mode \cite{li23-2} and multi-sector mode \cite{li23-3}.
Recent studies show the potential of BD-RIS in enabling novel applications, such as channel reciprocity attacks \cite{li22,wan24}, and prove the superiority of BD-RIS over D-RIS in single- and multi-user systems \cite{ner22,fan24}.

% Importance of mutual coupling
Although the benefits of BD-RIS have been shown from multiple aspects \cite{li23-1,she20,ner23-1,ner23-3,li23-2,li23-3,li22,wan24,ner22,fan24}, previous works are based on idealized assumptions neglecting the impact of \gls{em} mutual coupling between the RIS elements.
Mutual coupling refers to the \gls{em} interaction between the RIS elements.
When a RIS element is excited with a current, it generates an \gls{em} field that can induce currents in neighboring elements, thereby altering their excitation.
These mutual coupling effects at the RIS complicate the expression of the RIS-aided channel \cite{gra21,dir23,ner23-2}.
For this reason, mutual coupling is commonly neglected in the literature on D-RIS and BD-RIS.
Nevertheless, mutual coupling is an intrinsic phenomenon depending on the geometry of the RIS, which can affect the RIS behavior.
More importantly, not properly capturing it in the RIS-aided channel model during the RIS optimization can lead to performance degradation.
Therefore, modeling and managing mutual coupling between the RIS elements is essential for maximizing the benefits of RIS.

% Existing works on mutual coupling
Recent works have modeled and optimized RIS-aided systems accounting for mutual coupling.
Most literature focused on optimizing D-RIS with mutual coupling, in single-user systems \cite{qia21,per23} as well as in multi-user systems \cite{abr21,akr23,ma23,wij24}.
The impact of mutual coupling on the \gls{csi} acquisition has been investigated in \cite{zhe24-1,zhe24-3}, while RIS optimization in the presence of mutual coupling and imperfect \gls{csi} has been considered in \cite{pen24}.
In \cite{sem24}, decoupling networks at the RIS array have been proposed as a potential solution for handling mutual coupling.
Furthermore, in \cite{mur23,has24}, RIS-aided channels have been modeled accounting for the mutual coupling between the RIS elements and also the coupling induced by the presence of scattering objects.
These channel models have been validated through \gls{em} simulations and experiments in \cite{pet23} and \cite{zhe24-2}, respectively.
While the works \cite{qia21,per23,abr21,akr23,ma23,wij24,zhe24-1,zhe24-3,pen24,sem24,mur23,has24,pet23,zhe24-2} focused on the modeling and optimization of D-RIS, \cite{li24} analyzed the impact of mutual coupling on BD-RIS architectures.

% Gap
Despite recent efforts in studying RIS with mutual coupling, two open challenges can be identified.
\textit{First}, existing mutual coupling-aware optimization algorithms suffer from a high computational complexity caused by the high number of iterations required to converge, which becomes prohibitive when the number of RIS elements increases.
\textit{Second}, given their iterative nature, their convergence is guaranteed only to a local optimum.
%, and only subject to an appropriate choice of the step size.
Thus, an expression of the achievable channel gain of a RIS-aided system with mutual coupling remains unknown.
To solve these two challenges, in this study, we derive global optimal closed-form solutions to maximize the channel gain of BD-RIS-aided systems with mutual coupling, applicable to the fully- and tree-connected RIS architectures.
In addition, we characterize and compare the scaling laws of the channel gain in the presence and in the absence of mutual coupling.
Different from conventional efforts relying on iterative solutions converging to locally optimal solutions, we adopt a different strategy in this paper.
We first reformulate the objective function and the constraints in an equivalent way by introducing auxiliary variables.
We then derive an upper bound on the reformulated objective function and provide a closed-form solution for the optimization variables able to exactly achieve the upper bound.
The contributions of this study can be summarized as follows.

% Fully-connected
\textit{First}, we globally optimize in closed-form fully-connected RISs to maximize the channel gain in the presence of mutual coupling.
In addition, we provide the expression of the tight channel gain upper bound which is achievable by the proposed solution.
Numerical results are presented to verify the global optimality of the proposed solution.

% Tree-connected
\textit{Second}, we show that it is also possible to optimize tree-connected RISs through a different global optimal closed-form solution and provide the expression of the achievable channel gain.
We observe that tree-connected RISs achieve the same channel gain upper bound as fully-connected RISs in the presence of mutual coupling, while having a highly reduced circuit complexity.
Numerical results confirm that fully- and tree-connected RISs can be globally optimized in closed-form in the presence of mutual coupling, to exactly achieve the same channel gain upper bound.

% Scaling laws
\textit{Third}, we derive the scaling laws of the average channel gain obtained by a fully- or tree-connected RIS in the presence and in the absence of mutual coupling, under Rayleigh fading channels.
In the presence of mutual coupling, the scaling law is given as a closed-form function of the mutual coupling, while in the absence of mutual coupling, the scaling law is a function of the RIS antenna self-impedance and the number of RIS elements.
Numerical results show the accuracy of the derived scaling laws even for practical values of the number of RIS elements.

% Mutual coupling beneficial
\textit{Fourth}, we assess whether mutual coupling is detrimental or beneficial in improving the channel gain of a RIS-aided system.
To this end, we prove that the derived scaling law of the channel gain under Rayleigh fading channels with mutual coupling is always higher than with no mutual coupling, for any value of the mutual coupling.
Accordingly, we observe that stronger mutual coupling effects enable higher channel gains when they are accounted for by the proposed solutions.

\textit{Organization}:
In Section~\ref{sec:model}, we model a RIS-aided system with multiport network theory.
In Section~\ref{sec:FC}, we derive a global optimal closed-form solution to optimize fully-connected RISs with mutual coupling.
In Section~\ref{sec:TC}, we provide a global optimal closed-form solution for tree-connected RISs with mutual coupling, achieving the same performance as fully-connected RISs.
In Section~\ref{sec:scaling-laws}, we derive the scaling laws of the average channel gains achievable with fully- and tree-connected RISs, in the presence and in the absence of mutual coupling.
In Section~\ref{sec:impact}, we analytically assess the impact of mutual coupling on the average channel gain.
In Section~\ref{sec:results}, we provide numerical results to evaluate the channel gain of a BD-RIS-aided system with mutual coupling.
Finally, Section~\ref{sec:conclusion} concludes this work.

\textit{Notation}:
Vectors and matrices are denoted with bold lower and bold upper letters, respectively.
Scalars are represented with letters not in bold font.
$\Re\{a\}$, $\Im\{a\}$, $\vert a\vert$, $\arg(a)$, and $a^*$ refer to the real part, imaginary part, modulus, phase, and conjugate of a complex scalar $a$, respectively.
$[\mathbf{a}]_{i}$ and $\Vert\mathbf{a}\Vert_2$ refer to the $i$th element and $l_2$-norm of a vector $\mathbf{a}$, respectively.
$\mathbf{A}^{*}$, $\mathbf{A}^{T}$, $\mathbf{A}^{H}$, and $[\mathbf{A}]_{i,j}$ refer to the conjugate, transpose, conjugate transpose, and $(i,j)$th element of a matrix $\mathbf{A}$, respectively.
$\mathbb{R}$ and $\mathbb{C}$ denote the real and complex number sets, respectively.
$j=\sqrt{-1}$ denotes the imaginary unit.
$\mathbf{0}$ and $\mathbf{I}$ denote an all-zero matrix and an identity matrix, respectively, with appropriate dimensions.
$\mathcal{CN}(0,\sigma^2)$ denotes the distribution of a \gls{cscg} random variable whose real and imaginary parts are independent and Gaussian distributed with mean zero and variance $\sigma^2/2$.
$\mathcal{CN}(\mathbf{0},\mathbf{C})$ denotes the distribution of a \gls{cscg} random vector with mean vector $\mathbf{0}$ and covariance matrix $\mathbf{C}$.
diag$(a_{1},\ldots,a_{N})$ refers to a diagonal matrix with diagonal elements being $a_{1},\ldots,a_{N}$.

%%%%%%%%%%%%%%%%%%%%%%%%%%%%%%%%%%%%%%%%%%%%%%%%%%
\section{RIS-Aided Channel Model With\\Mutual Coupling}
\label{sec:model}

Consider a communication system between a single-antenna transmitter and a single-antenna receiver, aided by an $N_I$-element RIS.
This system can be modeled by using the multiport network theory \cite{she20,gra21,ner23-2}, and its wireless channel can be regarded as an $N$-port network, with $N=2+N_I$, as represented in Fig.~\ref{fig:diagram}.
We consider a \gls{siso} system to gain fundamental insights into the role of mutual coupling in RIS-aided wireless systems, as in related literature \cite{qia21,per23,li24}.
Nevertheless, the illustrated channel model can be readily extended to \gls{mimo} systems \cite{ner23-2}.

Following the multiport network theory \cite[Chapter 4]{poz11}, the $N$-port network representing the wireless channel is fully characterized by its impedance matrix $\mathbf{Z}\in\mathbb{C}^{N\times N}$, which can be partitioned as
\begin{equation}
\mathbf{Z}=
\begin{bmatrix}
z_{TT} & \mathbf{z}_{TI} & z_{TR}\\
\mathbf{z}_{IT} & \mathbf{Z}_{II} & \mathbf{z}_{IR}\\
z_{RT} & \mathbf{z}_{RI} & z_{RR}
\end{bmatrix}.\label{eq:Z}
\end{equation}
In \eqref{eq:Z}, $z_{TT}\in\mathbb{C}$ and $z_{RR}\in\mathbb{C}$ are the self-impedance of the antenna at the transmitter and receiver, respectively, and $\mathbf{Z}_{II}\in\mathbb{C}^{N_I\times N_I}$  refers to the impedance matrix of the antenna array at the RIS, whose diagonal entries represent the self-impedances of the RIS antennas while the off-diagonal entries represent the \gls{em} mutual coupling between the RIS antennas.
In addition, $\mathbf{z}_{IT}\in\mathbb{C}^{N_I\times 1}$, $\mathbf{z}_{RI}\in\mathbb{C}^{1\times N_I}$, and $z_{RT}\in\mathbb{C}$ are the transmission impedance matrices from the transmitter to RIS, from the RIS to receiver, and from the transmitter to receiver, respectively.
Similarly, $\mathbf{z}_{TI}\in\mathbb{C}^{1\times N_I}$, $\mathbf{z}_{IR}\in\mathbb{C}^{N_I\times 1}$, and $z_{TR}\in\mathbb{C}$ refer to the transmission impedance matrices from the RIS to transmitter, from the receiver to RIS, and from the receiver to transmitter, respectively.
In the case of reciprocal wireless channels, we have $\mathbf{z}_{TI}=\mathbf{z}_{IT}^T$, $\mathbf{z}_{IR}=\mathbf{z}_{RI}^T$, and $z_{TR}=z_{RT}$.

\begin{figure}[t]
\centering
\includegraphics[width=0.48\textwidth]{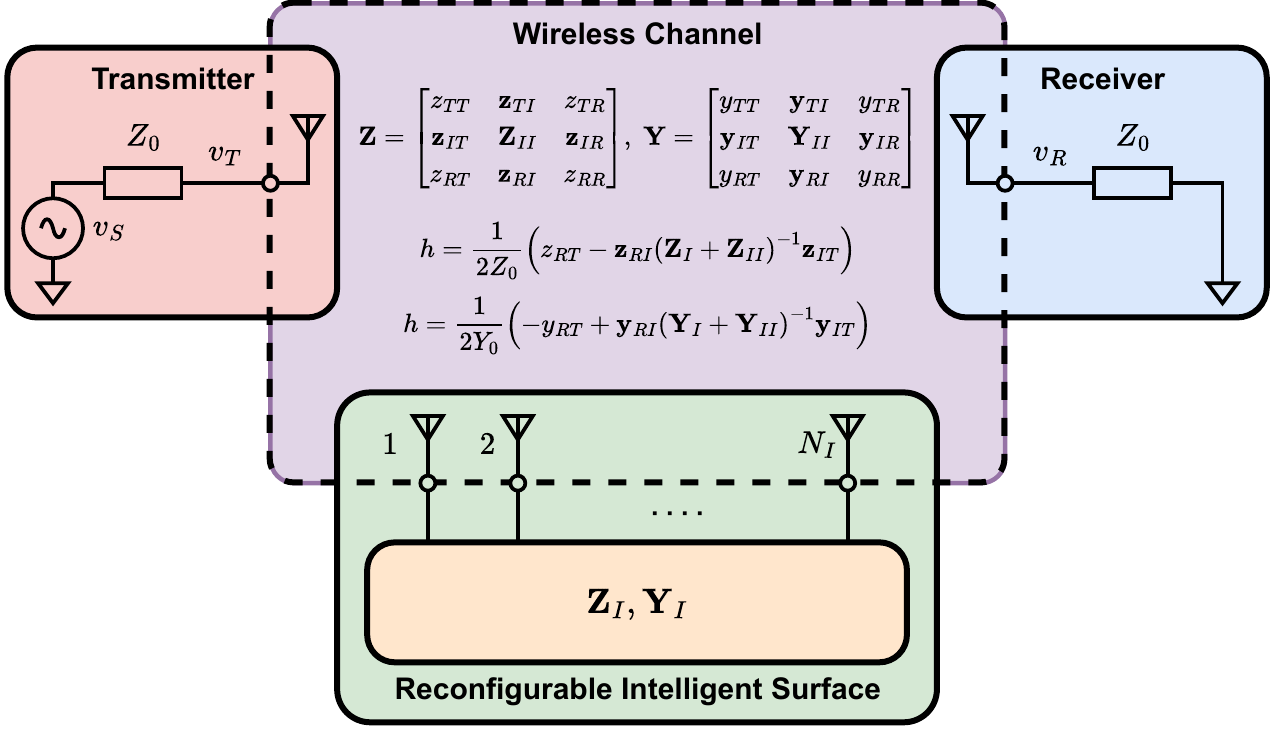}
\caption{RIS-aided system modeled with multiport network theory.}
\label{fig:diagram}
\end{figure}

The antenna at the transmitter is connected to a source voltage $v_{S}\in\mathbb{C}$ with series impedance $Z_0$, e.g., set to $Z_0=50\;\Omega$, and we denote the voltage at the transmitting antenna as $v_{T}\in\mathbb{C}$.
The antenna at the receiver is connected to a load impedance $Z_0$, and we denote the voltage at the receiving antenna as $v_{R}\in\mathbb{C}$.
At the RIS, the $N_I$ antennas are connected to an $N_I$-port reconfigurable impedance network with impedance matrix denoted as $\mathbf{Z}_{I}\in\mathbb{C}^{N_I\times N_I}$.
The impedance matrix $\mathbf{Z}_{I}$ is reconfigurable, and is a diagonal matrix for D-RIS, while it is generally not constrained to be diagonal for BD-RIS architectures.

To obtain a tractable expression of the channel $h\in\mathbb{C}$ relating the voltage $v_T$ (transmitted signal) and the voltage $v_R$ (received signal) through
\begin{equation}
v_R=hv_T+n,
\end{equation}
where $n$ is the \gls{awgn}, we make the following two assumptions, commonly considered in related literature \cite{she20,gra21,ner23-2}.
First, the transmission distances from the transmitter to RIS, from the RIS to receiver, and from the transmitter to receiver are assumed to be large enough such that we can neglect the effect of the feedback channels, i.e., we can consider $\mathbf{z}_{TI}=\mathbf{0}$, $\mathbf{z}_{IR}=\mathbf{0}$, and $z_{TR}=0$, which is also known as the unilateral approximation \cite{ivr10}.
Second, the antennas at the transmitter and receiver are assumed to be perfectly matched to $Z_0$, i.e., $z_{TT}=Z_0$ and $z_{RR}=Z_0$.
With these two assumptions, it is possible to show that the channel $h$ based on the $Z$-parameters representation reads as
\begin{equation}
h=\frac{1}{2Z_0}\left(z_{RT}-\mathbf{z}_{RI}\left(\mathbf{Z}_I+\mathbf{Z}_{II}\right)^{-1}\mathbf{z}_{IT}\right),\label{eq:h-Z}
\end{equation}
as derived in \cite{gra21,ner23-2}.

In addition to the $Z$-parameters, multiport network theory offers the $Y$-parameters as an equivalent representation for microwave multiport systems \cite[Chapter 4]{poz11}.
Based on this representation, the $N$-port network representing the wireless channel in Fig.~\ref{fig:diagram} is characterized by its admitacnce matrix $\mathbf{Y}=\mathbf{Z}^{-1}$, partitioned as
\begin{equation}
\mathbf{Y}=
\begin{bmatrix}
y_{TT} & \mathbf{y}_{TI} & y_{TR}\\
\mathbf{y}_{IT} & \mathbf{Y}_{II} & \mathbf{y}_{IR}\\
y_{RT} & \mathbf{y}_{RI} & y_{RR}
\end{bmatrix},\label{eq:Y}
\end{equation}
similarly to $\mathbf{Z}$ in \eqref{eq:Z}.
Considering the unilateral approximation and perfect matching at the transmitter and receiver, it has been shown in \cite{ner23-2} that the transmission admittance matrices $\mathbf{y}_{RI}\in\mathbb{C}^{1\times N_I}$, $\mathbf{y}_{IT}\in\mathbb{C}^{N_I\times 1}$, and $y_{RT}\in\mathbb{C}$ are given by
\begin{gather}
\mathbf{y}_{RI}=-\frac{\mathbf{z}_{RI}\mathbf{Z}_{II}^{-1}}{Z_0},\;\mathbf{y}_{IT}=-\frac{\mathbf{Z}_{II}^{-1}\mathbf{z}_{IT}}{Z_0},\label{eq:yRI-yIT}\\
y_{RT}=\frac{1}{Z_0^2}\left(-z_{RT}+\mathbf{z}_{RI}\mathbf{Z}_{II}^{-1}\mathbf{z}_{IT}\right),\label{eq:yRT}
\end{gather}
and the corresponding channel model reads as
\begin{equation}
h=\frac{1}{2Y_0}\left(-y_{RT}+\mathbf{y}_{RI}\left(\mathbf{Y}_I+\mathbf{Y}_{II}\right)^{-1}\mathbf{y}_{IT}\right),\label{eq:h-Y}
\end{equation}
where $Y_0=Z_0^{-1}$, $\mathbf{Y}_{I}=\mathbf{Z}_{I}^{-1}$, $\mathbf{Y}_{II}=\mathbf{Z}_{II}^{-1}$.
Although \eqref{eq:h-Y} is equal to \eqref{eq:h-Z} given the equivalence between $Z$- and $Y$-parameters (as shown in \cite{ner23-2}), the two models are best suited for their respective use cases.
Specifically, the $Z$-parameters model in \eqref{eq:h-Z} is suitable for fully-connected RISs with a full matrix $\mathbf{Z}_I$.
For the sake of illustration, we report the diagram of a $4$-element fully-connected RIS in Fig.~\ref{fig:fully-tree}(a), showing that each antenna port is connected to ground and to all other ports through a tunable impedance \cite{she20}.
Meanwhile, \eqref{eq:h-Y} is suitable for tree-connected RISs, where the admittance matrix $\mathbf{Y}_I$ can explicitly capture the specific circuit topology of such RIS architectures.
The diagram of a $4$-element tree-connected RIS is shown in Fig.~\ref{fig:fully-tree}(b), where the antenna ports are interconnected to each other such that the resulting topology can be represented with an acyclic and connected graph, i.e., a tree \cite{ner23-1}.
%Given the equivalence between $Z$- and $Y$-parameters, \eqref{eq:h-Y} is equal to \eqref{eq:h-Z}, and the preferred model can be chosen based on the problem at hand.
In the following sections, we derive global optimal solutions to maximize the channel gain for fully- and tree-connected RISs, using the $Z$- and $Y$-parameters representations, respectively.
The channel gain maximization problem is considered since it is equivalent to maximizing the \gls{snr} at the receiver.

%%%%%%%%%%%%%%%%%%%%%%%%%%%%%%%%%%%%%%%%%%%%%%%%%%
\section{Optimizing Fully-Connected RIS With\\Mutual Coupling}
\label{sec:FC}

In this section, we provide a global optimal closed-form solution for the channel gain maximization problem in the presence of a lossless and reciprocal fully-connected RIS.
To this end, we formulate the optimization problem based on the $Z$-parameters representation.

\begin{figure}[t]
\centering
\includegraphics[width=0.46\textwidth]{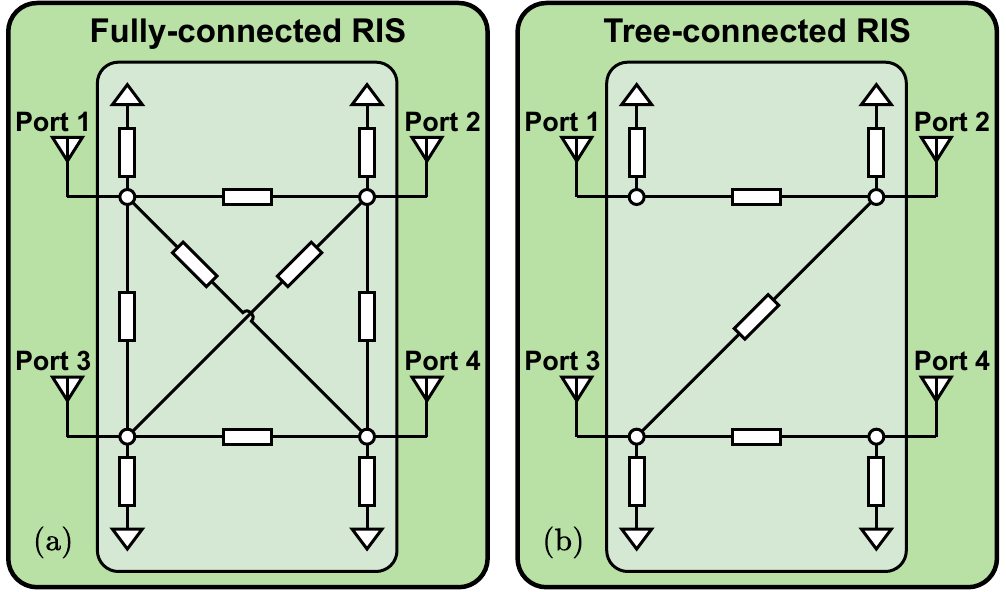}
\caption{Circuit topology of a (a) fully-connected RIS and (b) tree-connected RIS, with $N_I=4$ elements.}
\label{fig:fully-tree}
\end{figure}

\subsection{Problem Formulation and Global Optimal Solution}

For a lossless fully-connected RIS, $\mathbf{Z}_I$ is a purely imaginary matrix, i.e., $\mathbf{Z}_I=j\mathbf{X}_{I}$, where $\mathbf{X}_{I}\in\mathbb{R}^{N_I\times N_I}$ is the reactance matrix of the RIS.
Besides, the reciprocity of the microwave network implementing the fully-connected RIS imposes that $\mathbf{X}_I$ is symmetric, i.e., $\mathbf{X}_{I}=\mathbf{X}_{I}^{T}$.
Given these constraints, the channel gain $\vert h\vert^2$ maximization problem is given by
\begin{align}
\underset{\mathbf{X}_{I}}{\mathsf{\mathrm{max}}}\;\;
&\frac{1}{4Z_0^2}\left\vert z_{RT}-\mathbf{z}_{RI}\left(j\mathbf{X}_{I}+\mathbf{Z}_{II}\right)^{-1}\mathbf{z}_{IT}\right\vert^{2}\label{eq:p1-FC-1}\\
\mathsf{\mathrm{s.t.}}\;\;\;
&\mathbf{X}_{I}=\mathbf{X}_{I}^{T},\label{eq:p1-FC-2}
\end{align}
by using the channel model in \eqref{eq:h-Z}.
To solve \eqref{eq:p1-FC-1}-\eqref{eq:p1-FC-2}, we assume that the mutual coupling matrix $\mathbf{Z}_{II}$ is known.
This matrix depends purely on the array geometry, i.e., antenna element type and inter-element distances, and is independent of the channels between the RIS and the transmitter and receiver.
Thus, it can be estimated offline, before the RIS is deployed.
We also assume perfect \gls{csi}, which can be acquired using the channel estimation protocols developed in previous literature \cite{zhe24-1,zhe24-3}, suitably adapted for BD-RIS.

A global optimal closed-form solution for this channel gain maximization problem without mutual coupling at the RIS, i.e., with $\mathbf{Z}_{II}=Z_0\mathbf{I}$, has been proposed in \cite{ner22}.
Thus, we globally solve \eqref{eq:p1-FC-1}-\eqref{eq:p1-FC-2} by ``diagonalizing'' $\mathbf{Z}_{II}$ such that the solution in \cite{ner22} can be directly adopted.
To this end, we first introduce the following result giving a useful property of the mutual coupling matrix $\mathbf{Z}_{II}$.

\begin{proposition}
For a reciprocal and lossy $N$-port network with impedance matrix $\mathbf{Z}\in\mathbb{C}^{N\times N}$ and admittance matrix $\mathbf{Y}=\mathbf{Z}^{-1}$, the matrices $\Re\{\mathbf{Z}\}$ and $\Re\{\mathbf{Y}\}$ are positive semi-definite.
\label{pro:semi-definite}
\end{proposition}
\begin{proof}Please refer to Appendix~A.\end{proof}

By exploiting Proposition~\ref{pro:semi-definite} and assuming $\Re\{\mathbf{Z}_{II}\}$ to be invertible, we define the auxiliary variable $\bar{\mathbf{X}}_I\in\mathbb{R}^{N_I\times N_I}$ as
\begin{equation}
\bar{\mathbf{X}}_{I}=Z_0\Re\{\mathbf{Z}_{II}\}^{-1/2}(\mathbf{X}_{I}+\Im\{\mathbf{Z}_{II}\})\Re\{\mathbf{Z}_{II}\}^{-1/2},\label{eq:XI-bar}
\end{equation}
which is a real matrix since $\mathbf{X}_{I}$, $\Im\{\mathbf{Z}_{II}\}$, and $\Re\{\mathbf{Z}_{II}\}^{-1/2}$ are real matrices.
Note that $\Re\{\mathbf{Z}_{II}\}^{-1/2}$ is a real matrix since $\Re\{\mathbf{Z}_{II}\}$ is positive semi-definite following Proposition~\ref{pro:semi-definite} and its invertibility.
The rationale behind this auxiliary variable is that it allows to ``diagonalize'' the mutual coupling, similar to the effect of the decoupling network used in \cite{sem24}.
Thus, by substituting
\begin{equation}
\mathbf{X}_{I}=\frac{1}{Z_0}\Re\{\mathbf{Z}_{II}\}^{1/2}\bar{\mathbf{X}}_{I}\Re\{\mathbf{Z}_{II}\}^{1/2}-\Im\{\mathbf{Z}_{II}\},\label{eq:XI}
\end{equation}
which follows from \eqref{eq:XI-bar}, into \eqref{eq:p1-FC-1}, problem \eqref{eq:p1-FC-1}-\eqref{eq:p1-FC-2} can be equivalently rewritten as
\begin{align}
\underset{\mathbf{X}_{I}}{\mathsf{\mathrm{max}}}\;\;
&\frac{1}{4Z_0^2}\left\vert z_{RT}-\mathbf{z}_{RI}\Re\{\mathbf{Z}_{II}\}^{-1/2}\sqrt{Z_0}\right.\notag\\
&\left.\times\left(j\bar{\mathbf{X}}_{I}+Z_0\mathbf{I}\right)^{-1}\sqrt{Z_0}\Re\{\mathbf{Z}_{II}\}^{-1/2}\mathbf{z}_{IT}\right\vert^{2}\label{eq:p2-FC-1}\\
\mathsf{\mathrm{s.t.}}\;\;\;
&\eqref{eq:XI-bar},\;
\mathbf{X}_{I}=\mathbf{X}_{I}^{T}.\label{eq:p2-FC-2}
\end{align}
Remarkably, constraint \eqref{eq:p2-FC-2} indicates that  $\bar{\mathbf{X}}_{I}$ can be an arbitrary symmetric matrix since $\mathbf{Z}_{II}$ is symmetric following the reciprocity of the mutual coupling effects.
Thus, we can solve problem \eqref{eq:p2-FC-1}-\eqref{eq:p2-FC-2} for $\bar{\mathbf{X}}_{I}$ and then find the optimal $\mathbf{X}_{I}$ through \eqref{eq:XI}.
To this end, we equivalently rewrite problem \eqref{eq:p2-FC-1}-\eqref{eq:p2-FC-2} as
\begin{align}
\underset{\bar{\mathbf{X}}_{I}}{\mathsf{\mathrm{max}}}\;\;
&\frac{1}{4Z_0^2}\left\vert z_{RT}-\bar{\mathbf{z}}_{RI}\left(j\bar{\mathbf{X}}_{I}+Z_0\mathbf{I}\right)^{-1}\bar{\mathbf{z}}_{IT}\right\vert^{2}\label{eq:p3-FC-1}\\
\mathsf{\mathrm{s.t.}}\;\;\;
&\bar{\mathbf{X}}_{I}=\bar{\mathbf{X}}_{I}^{T},\label{eq:p3-FC-2}
\end{align}
where we introduced $\bar{\mathbf{z}}_{RI}\in\mathbb{C}^{1\times N_I}$ and $\bar{\mathbf{z}}_{IT}\in\mathbb{C}^{N_I\times 1}$ as
\begin{align}
\bar{\mathbf{z}}_{RI}&=\mathbf{z}_{RI}\Re\{\mathbf{Z}_{II}\}^{-1/2}\sqrt{Z_0},\label{eq:zRI-bar}\\
\bar{\mathbf{z}}_{IT}&=\sqrt{Z_0}\Re\{\mathbf{Z}_{II}\}^{-1/2}\mathbf{z}_{IT}.\label{eq:zIT-bar}
\end{align}
Interestingly, \eqref{eq:p3-FC-1} is equivalent to the channel expression in the presence of a RIS with reactance matrix $\bar{\mathbf{X}}_{I}$ and no mutual coupling, and a similar expression has been obtained through the use of a decoupling network in \cite{sem24}.
To solve problem \eqref{eq:p3-FC-1}-\eqref{eq:p3-FC-2} with the solution proposed in \cite{ner22}, we map the $Z$-parameter based channels into the $S$-parameter based ones \cite{ner23-2} by defining
\begin{gather}
\bar{\mathbf{s}}_{RI}=\frac{\bar{\mathbf{z}}_{RI}}{2Z_0},\;
\bar{\mathbf{s}}_{IT}=\frac{\bar{\mathbf{z}}_{IT}}{2Z_0},\label{eq:sRI-sIT-bar-FC}\\
\bar{s}_{RT}=\frac{1}{2Z_0}\left(z_{RT}-\frac{\bar{\mathbf{z}}_{RI}\bar{\mathbf{z}}_{IT}}{2Z_0}\right),\label{eq:sRT-bar-FC}\\
\bar{\boldsymbol{\Theta}}=\left(j\bar{\mathbf{X}}_I+Z_0\mathbf{I}\right)^{-1}\left(j\bar{\mathbf{X}}_I-Z_0\mathbf{I}\right),\label{eq:T-bar-FC}
\end{gather}
and equivalently rewrite \eqref{eq:p3-FC-1}-\eqref{eq:p3-FC-2} as
\begin{align}
\underset{\bar{\mathbf{X}}_{I}}{\mathsf{\mathrm{max}}}\;\;
&\left\vert \bar{s}_{RT}+\bar{\mathbf{s}}_{RI}\bar{\boldsymbol{\Theta}}\bar{\mathbf{s}}_{IT}\right\vert^{2}\label{eq:p4-FC-1}\\
\mathsf{\mathrm{s.t.}}\;\;\;
&\eqref{eq:T-bar-FC},\;
\bar{\mathbf{X}}_{I}=\bar{\mathbf{X}}_{I}^{T}.\label{eq:p4-FC-2}
\end{align}
Noticing that constraint \eqref{eq:p4-FC-2} implies that  $\bar{\boldsymbol{\Theta}}$ can be an arbitrary unitary and symmetric matrix (since $\bar{\mathbf{X}}_{I}$ is a real matrix given the positive-definiteness of $\Re\{\mathbf{Z}_{II}\}$), we can equivalently transform \eqref{eq:p4-FC-1}-\eqref{eq:p4-FC-2} into
\begin{align}
\underset{\bar{\boldsymbol{\Theta}}}{\mathsf{\mathrm{max}}}\;\;
&\left\vert \bar{s}_{RT}+\bar{\mathbf{s}}_{RI}\bar{\boldsymbol{\Theta}}\bar{\mathbf{s}}_{IT}\right\vert^{2}\label{eq:p5-FC-1}\\
\mathsf{\mathrm{s.t.}}\;\;\;
&\bar{\boldsymbol{\Theta}}^H\bar{\boldsymbol{\Theta}}=\mathbf{I},\;\bar{\boldsymbol{\Theta}}=\bar{\boldsymbol{\Theta}}^{T},\label{eq:p5-FC-2}
\end{align}
and compute $\bar{\mathbf{X}}_{I}$ based on the optimal $\bar{\boldsymbol{\Theta}}$ by inverting \eqref{eq:T-bar-FC}.

From problem \eqref{eq:p5-FC-1}-\eqref{eq:p5-FC-2}, the global optimal reactance matrix of the RIS $\mathbf{X}_I$ can be found according to the following three steps.
\textit{First}, problem \eqref{eq:p5-FC-1}-\eqref{eq:p5-FC-2} is globally solved through the global optimal closed-form solution for $\bar{\boldsymbol{\Theta}}$ proposed in \cite{ner22}, returning $\bar{\boldsymbol{\Theta}}^\star$.
\textit{Second}, the global optimal $\bar{\mathbf{X}}_{I}$, denoted as $\bar{\mathbf{X}}_{I}^\star$, is found from $\bar{\boldsymbol{\Theta}}^\star$ by inverting \eqref{eq:T-bar-FC}, i.e.,
%$\bar{\mathbf{X}}_{I}^\star=-jZ_0(\mathbf{I}+\bar{\boldsymbol{\Theta}}^\star)(\mathbf{I}-\bar{\boldsymbol{\Theta}}^\star)^{-1}$.
\begin{equation}
\bar{\mathbf{X}}_{I}^\star=-jZ_0(\mathbf{I}+\bar{\boldsymbol{\Theta}}^\star)(\mathbf{I}-\bar{\boldsymbol{\Theta}}^\star)^{-1}.
\end{equation}
\textit{Third}, the global optimal $\mathbf{X}_{I}$ is obtained from $\bar{\mathbf{X}}_{I}^\star$ via \eqref{eq:XI}.
The computational complexity of optimizing a fully-connected RIS is driven by the complexity of solving problem \eqref{eq:p5-FC-1}-\eqref{eq:p5-FC-2} with the solution proposed in \cite{ner22}.
Since the solution in \cite{ner22} requires the computation of the eigenvalue decomposition of an $N_I\times N_I$ matrix, its complexity scales with $\mathcal{O}(N_I^3)$.

\subsection{Channel Gain Upper Bound}

The proposed solution is proved to be global optimal since it can exactly achieve the upper bound on the channel gain $\vert \bar{s}_{RT}+\bar{\mathbf{s}}_{RI}\bar{\boldsymbol{\Theta}}\bar{\mathbf{s}}_{IT}\vert^{2}$, which is given by
\begin{equation}
\left\vert h^\star\right\vert^2
=\left(\left\vert \bar{s}_{RT}\right\vert+\left\Vert\bar{\mathbf{s}}_{RI}\right\Vert_2\left\Vert\bar{\mathbf{s}}_{IT}\right\Vert_2\right)^{2},\label{eq:UB-FC-1}
\end{equation}
following the triangle inequality, the sub multiplicity property of the norm, and that $\Vert\bar{\boldsymbol{\Theta}}\Vert_2=1$ for any unitary $\bar{\boldsymbol{\Theta}}$.
To write this upper bound as a function of the channels $z_{RT}$, $\mathbf{z}_{RI}$, and $\mathbf{z}_{IT}$, and the mutual coupling matrix $\mathbf{Z}_{II}$, we substitute \eqref{eq:sRI-sIT-bar-FC} and \eqref{eq:sRT-bar-FC} into \eqref{eq:UB-FC-1} to obtain
\begin{equation}
\left\vert h^\star\right\vert^2
=\frac{1}{4Z_0^2}
\left(\left\vert z_{RT}-\frac{\bar{\mathbf{z}}_{RI}\bar{\mathbf{z}}_{IT}}{2Z_0}\right\vert+\frac{\left\Vert\bar{\mathbf{z}}_{RI}\right\Vert_2\left\Vert\bar{\mathbf{z}}_{IT}\right\Vert_2}{2Z_0}\right)^{2},\label{eq:UB-FC-2}
\end{equation}
and substitute \eqref{eq:zRI-bar} and \eqref{eq:zIT-bar} into \eqref{eq:UB-FC-2} to get
\begin{multline}
\left\vert h^\star\right\vert^2
=\frac{1}{4Z_0^2}
\left(\left\vert z_{RT}-\frac{1}{2}\mathbf{z}_{RI}\Re\{\mathbf{Z}_{II}\}^{-1}\mathbf{z}_{IT}\right\vert\right.\\
\left.+\frac{1}{2}\left\Vert\mathbf{z}_{RI}\Re\{\mathbf{Z}_{II}\}^{-1/2}\right\Vert_2\left\Vert\Re\{\mathbf{Z}_{II}\}^{-1/2}\mathbf{z}_{IT}\right\Vert_2\right)^{2},\label{eq:UB-FC-3}
\end{multline}
being an explicit function of the channels and the mutual coupling.
In \eqref{eq:UB-FC-3}, the term $\vert z_{RT}-\mathbf{z}_{RI}\Re\{\mathbf{Z}_{II}\}^{-1}\mathbf{z}_{IT}/2\vert$ represents the strength of the direct link $z_{RT}$ and the structural scattering (or specular reflection) of the RIS $-\mathbf{z}_{RI}\Re\{\mathbf{Z}_{II}\}^{-1}\mathbf{z}_{IT}/2$.
In addition, $\Vert\mathbf{z}_{RI}\Re\{\mathbf{Z}_{II}\}^{-1/2}\Vert_2$ and $\Vert\Re\{\mathbf{Z}_{II}\}^{-1/2}\mathbf{z}_{IT}\Vert_2$ are the strengths of the effective RIS-receiver and transmitter-RIS channels, respectively.
Note that \eqref{eq:UB-FC-3} is an explicit expression of the maximum channel gain achievable by a fully-connected RIS with mutual coupling, which is crucial to conducting performance analysis of RIS-aided systems with mutual coupling.

The fully-connected RIS is the most flexible BD-RIS architecture, enabling the highest performance at the cost of a higher circuit complexity.
Specifically, in a fully-connected RIS, there are $N_I(N_I+1)/2$ tunable impedance components interconnecting all the RIS elements to each other and to ground \cite{she20}.
To decrease the circuit complexity, the tree-connected RIS has been proposed, which is the least complex BD-RIS architecture achieving the same performance as the fully-connected RIS without mutual coupling, while including only $2N_I-1$ tunable impedance components \cite{ner23-1}.
In the following, we show that tree-connected RIS can achieve the same performance as the fully-connected RIS also in the presence of mutual coupling.

%%%%%%%%%%%%%%%%%%%%%%%%%%%%%%%%%%%%%%%%%%%%%%%%%%
\section{Optimizing Tree-Connected RIS With\\Mutual Coupling}
\label{sec:TC}

In this section, we provide a global optimal closed-form solution for the channel gain maximization problem in the presence of a lossless and reciprocal tree-connected RIS.
To this end, we exploit the $Y$-parameters channel representation since the constraints for the tree-connected RIS are captured on its admittance matrix $\mathbf{Y}_{I}$, as explained in \cite{ner23-1}.

\subsection{Problem Formulation and Global Optimal Solution}

For a lossless tree-connected RIS, its admittance matrix $\mathbf{Y}_I$ is purely imaginary, i.e., $\mathbf{Y}_I=j\mathbf{B}_{I}$, where $\mathbf{B}_{I}\in\mathbb{R}^{N_I\times N_I}$ is the susceptance matrix of the RIS.
For a reciprocal RIS, $\mathbf{B}_I$ is symmetric, i.e., $\mathbf{B}_{I}=\mathbf{B}_{I}^{T}$.
In addition, $\mathbf{B}_I$ is also subject to the constraints given by the specific tree-connected RIS architecture considered.
Since numerous tree-connected RIS architectures are possible, we consider the tridiagonal RIS architecture in the following, having $[\mathbf{B}_{I}]_{i,j}=\mathbf{0}$ if $\left\vert i-j\right\vert>1$ \cite{ner23-1}.
Nevertheless, the following discussion is valid for every tree-connected RIS architecture.
For a tridiagonal RIS, the channel gain $\vert h\vert^2$ maximization problem writes as
\begin{align}
\underset{\mathbf{B}_{I}}{\mathsf{\mathrm{max}}}\;\;
&\frac{1}{4Y_0^2}\left\vert y_{RT}-\mathbf{y}_{RI}\left(j\mathbf{B}_{I}+\mathbf{Y}_{II}\right)^{-1}\mathbf{y}_{IT}\right\vert^{2}\label{eq:p1-TC-1}\\
\mathsf{\mathrm{s.t.}}\;\;\;
&\mathbf{B}_{I}=\mathbf{B}_{I}^{T},\;[\mathbf{B}_{I}]_{i,j}=\mathbf{0}\text{ if }\left\vert i-j\right\vert>1,\label{eq:p1-TC-2}
\end{align}
by employing the channel model in \eqref{eq:h-Y}.

Similar to the case with fully-connected RIS, a global optimal closed-form solution for this channel gain maximization problem without mutual coupling at the RIS, i.e., with $\mathbf{Y}_{II}=Y_0\mathbf{I}$, has been proposed in \cite{ner23-1}.
Thus, problem \eqref{eq:p1-TC-1}-\eqref{eq:p1-TC-2} can be solved by ``diagonalizing'' $\mathbf{Y}_{II}$ such that the solution proposed in \cite{ner23-1} can be directly adopted.
To this end, we introduce the auxiliary variable $\bar{\mathbf{B}}_I\in\mathbb{R}^{N_I\times N_I}$ as
\begin{equation}
\bar{\mathbf{B}}_{I}=Y_0\Re\{\mathbf{Y}_{II}\}^{-1/2}(\mathbf{B}_{I}+\Im\{\mathbf{Y}_{II}\})\Re\{\mathbf{Y}_{II}\}^{-1/2},\label{eq:BI-bar}
\end{equation}
which is a real matrix following the positive definiteness of $\Re\{\mathbf{Y}_{II}\}$ (see Proposition~\ref{pro:semi-definite}), giving
\begin{equation}
\mathbf{B}_{I}=\frac{1}{Y_0}\Re\{\mathbf{Y}_{II}\}^{1/2}\bar{\mathbf{B}}_{I}\Re\{\mathbf{Y}_{II}\}^{1/2}-\Im\{\mathbf{Y}_{II}\}.\label{eq:BI}
\end{equation}
Thus, by substituting \eqref{eq:BI} into \eqref{eq:p1-TC-1}, problem \eqref{eq:p1-TC-1}-\eqref{eq:p1-TC-2} becomes
\begin{align}
\underset{\mathbf{B}_{I}}{\mathsf{\mathrm{max}}}\;\;
&\frac{1}{4Y_0^2}\left\vert y_{RT}-\mathbf{y}_{RI}\Re\{\mathbf{Y}_{II}\}^{-1/2}\sqrt{Y_0}\right.\notag\\
&\left.\times\left(j\bar{\mathbf{B}}_{I}+Y_0\mathbf{I}\right)^{-1}\sqrt{Y_0}\Re\{\mathbf{Y}_{II}\}^{-1/2}\mathbf{y}_{IT}\right\vert^{2}\label{eq:p2-TC-1}\\
\mathsf{\mathrm{s.t.}}\;\;\;
&\eqref{eq:BI-bar},\;
\mathbf{B}_{I}=\mathbf{B}_{I}^{T},\;[\mathbf{B}_{I}]_{i,j}=\mathbf{0}\text{ if }\left\vert i-j\right\vert>1,\label{eq:p2-TC-2}
\end{align}
which can be rewritten as
\begin{align}
\underset{\mathbf{B}_{I}}{\mathsf{\mathrm{max}}}\;\;
&\frac{1}{4Y_0^2}\left\vert y_{RT}-\bar{\mathbf{y}}_{RI}\left(j\bar{\mathbf{B}}_{I}+Y_0\mathbf{I}\right)^{-1}\bar{\mathbf{y}}_{IT}\right\vert^{2}\label{eq:p3-TC-1}\\
\mathsf{\mathrm{s.t.}}\;\;\;
&\eqref{eq:BI-bar},\;
\mathbf{B}_{I}=\mathbf{B}_{I}^{T},\;[\mathbf{B}_{I}]_{i,j}=\mathbf{0}\text{ if }\left\vert i-j\right\vert>1,\label{eq:p3-TC-2}
\end{align}
by introducing
\begin{align}
\bar{\mathbf{y}}_{RI}&=\mathbf{y}_{RI}\Re\{\mathbf{Y}_{II}\}^{-1/2}\sqrt{Y_0},\label{eq:yRI-bar}\\
\bar{\mathbf{y}}_{IT}&=\sqrt{Y_0}\Re\{\mathbf{Y}_{II}\}^{-1/2}\mathbf{y}_{IT}.\label{eq:yIT-bar}
\end{align}
To globally solve problem \eqref{eq:p3-TC-1}-\eqref{eq:p3-TC-2}, we introduce
\begin{gather}
\bar{\mathbf{s}}_{RI}=-\frac{\bar{\mathbf{y}}_{RI}}{2Y_0},\;
\bar{\mathbf{s}}_{IT}=-\frac{\bar{\mathbf{y}}_{IT}}{2Y_0},\label{eq:sRI-sIT-bar-TC}\\
\bar{s}_{RT}=-\frac{1}{2Y_0}\left(y_{RT}-\frac{\bar{\mathbf{y}}_{RI}\bar{\mathbf{y}}_{IT}}{2Y_0}\right),\label{eq:sRT-bar-TC}\\
\bar{\boldsymbol{\Theta}}=\left(Y_0\mathbf{I}+j\bar{\mathbf{B}}_I\right)^{-1}\left(Y_0\mathbf{I}-j\bar{\mathbf{B}}_I\right),\label{eq:T-bar-TC}
\end{gather}
such that it can be transformed into
\begin{align}
\underset{\mathbf{B}_{I}}{\mathsf{\mathrm{max}}}\;\;
&\left\vert \bar{s}_{RT}+\bar{\mathbf{s}}_{RI}\bar{\boldsymbol{\Theta}}\bar{\mathbf{s}}_{IT}\right\vert^{2}\label{eq:p4-TC-1}\\
\mathsf{\mathrm{s.t.}}\;\;\;
&\eqref{eq:T-bar-TC},\;\eqref{eq:BI-bar},\;
\mathbf{B}_{I}=\mathbf{B}_{I}^{T},\;[\mathbf{B}_{I}]_{i,j}=\mathbf{0}\text{ if }\left\vert i-j\right\vert>1.\label{eq:p4-TC-2}
\end{align}

To find a global optimal solution of \eqref{eq:p4-TC-1}-\eqref{eq:p4-TC-2} that exactly achieves the performance upper bound in \eqref{eq:UB-FC-1}, we need to find a $\bar{\mathbf{\Theta}}$ such that
\begin{equation}
e^{j\varphi_{RT}}\hat{\mathbf{s}}_{RI}^{H}=\bar{\boldsymbol{\Theta}}\hat{\mathbf{s}}_{IT},\label{eq:cond1}
\end{equation}
where $\varphi_{RT}\in\mathbb{C}$, $\hat{\mathbf{s}}_{RI}\in\mathbb{C}^{1\times N_I}$, and $\hat{\mathbf{s}}_{IT}\in\mathbb{C}^{N_I\times1}$ are given by
\begin{equation}
\varphi_{RT}=\arg\left(\bar{s}_{RT}\right),\;
\hat{\mathbf{s}}_{RI}=\frac{\bar{\mathbf{s}}_{RI}}{\left\Vert\bar{\mathbf{s}}_{RI}\right\Vert_2},\;
\hat{\mathbf{s}}_{IT}=\frac{\bar{\mathbf{s}}_{IT}}{\left\Vert\bar{\mathbf{s}}_{IT}\right\Vert_2}.
\end{equation}
By expressing $\bar{\boldsymbol{\Theta}}$ as in \eqref{eq:T-bar-TC}, condition \eqref{eq:cond1} can be equivalently rewritten as
%\begin{equation}
%\left(Y_0\mathbf{I}+j\bar{\mathbf{B}}_I\right)e^{j\varphi_{RT}}\hat{\mathbf{s}}_{RI}^{H}=\left(Y_0\mathbf{I}-j\bar{\mathbf{B}}_I\right)\hat{\mathbf{s}}_{IT},\label{eq:cond2}
%\end{equation}
\begin{equation}
\bar{\mathbf{B}}_I\bar{\boldsymbol{\alpha}}=\bar{\boldsymbol{\beta}},\label{eq:cond2}
\end{equation}
where we introduced $\bar{\boldsymbol{\alpha}}\in\mathbb{C}^{N_I\times1}$
and $\bar{\boldsymbol{\beta}}\in\mathbb{C}^{N_I\times1}$ as
\begin{align}
\bar{\boldsymbol{\alpha}}&=j\left(\hat{\mathbf{s}}_{IT}+e^{j\varphi_{RT}}\hat{\mathbf{s}}_{RI}^{H}\right),\\
\bar{\boldsymbol{\beta}}&=Y_{0}\left(\hat{\mathbf{s}}_{IT}-e^{j\varphi_{RT}}\hat{\mathbf{s}}_{RI}^{H}\right).
\end{align}
Furthermore, by using \eqref{eq:BI-bar} to express $\bar{\mathbf{B}}_{I}$, condition \eqref{eq:cond2} is achieved if and only if
%\begin{equation}
%\left(\mathbf{B}_{I}+\Im\{\mathbf{Y}_{II}\}\right)\Re\{\mathbf{Y}_{II}\}^{-1/2}\bar{\boldsymbol{\alpha}}=\frac{1}{Y_0}\Re\{\mathbf{Y}_{II}\}^{1/2}\bar{\boldsymbol{\beta}},\label{eq:cond3}
%\end{equation}
\begin{equation}
\mathbf{B}_I\boldsymbol{\alpha}=\boldsymbol{\beta},\label{eq:cond3}
\end{equation}
where $\boldsymbol{\alpha}\in\mathbb{C}^{N_I\times1}$ and $\boldsymbol{\beta}\in\mathbb{C}^{N_I\times1}$ are introduced as
\begin{align}
\boldsymbol{\alpha}&=\Re\{\mathbf{Y}_{II}\}^{-1/2}\bar{\boldsymbol{\alpha}},\\
\boldsymbol{\beta}&=\frac{1}{Y_0}\Re\{\mathbf{Y}_{II}\}^{1/2}\bar{\boldsymbol{\beta}}-\Im\{\mathbf{Y}_{II}\}\Re\{\mathbf{Y}_{II}\}^{-1/2}\bar{\boldsymbol{\alpha}}.
\end{align}
Note that condition \eqref{eq:cond3} is a system of $N_I$ linear equations in complex coefficients, i.e., the entries of $\boldsymbol{\alpha}$ and $\boldsymbol{\beta}$, with $2N_I-1$ real unknown, i.e., the entries of $\mathbf{B}_I$ that are not constrained to zero.
Remarkably, it is possible to prove that the system in \eqref{eq:cond3} has exactly one solution in general \cite{ner23-1}, and can be solved through the algorithm proposed in \cite{ner23-1} for given $\boldsymbol{\alpha}$ and $\boldsymbol{\beta}$.
The computational complexity of optimizing a tree-connected RIS is driven by the complexity of solving the system of equations in \eqref{eq:cond3} with the solution proposed in \cite{ner23-1}.
Since the system of equations in \eqref{eq:cond3} includes $2N_I-1$ independent equations with real coefficients and $2N_I-1$ real unknown \cite{ner23-1}, and solving such a system requires the inversion of a $(2N_I-1)\times(2N_I-1)$ matrix, the complexity of its solution scales with $\mathcal{O}(N_I^3)$.

\subsection{Channel Gain Upper Bound}

With the global optimal solution of the susceptance matrix $\mathbf{B}_I^\star$ obtained by solving the system \eqref{eq:cond3} through the algorithm proposed in \cite{ner23-1}, the performance upper bound \eqref{eq:UB-FC-1} is exactly achieved.
In the case of the $Y$-parameters representation,
%In problem \eqref{eq:p4-TC-1}-\eqref{eq:p4-TC-4}, the channel gain $\vert \bar{s}_{RT}+\bar{\mathbf{s}}_{RI}\bar{\boldsymbol{\Theta}}\bar{\mathbf{s}}_{IT}\vert^{2}$ is upper bounded by
%\begin{equation}
%\left\vert h^\star\right\vert^2
%=\left(\left\vert \bar{s}_{RT}\right\vert+\left\Vert\bar{\mathbf{s}}_{RI}\right\Vert_2\left\Vert\bar{\mathbf{s}}_{IT}\right\Vert_2\right)^{2},\label{eq:UB-TC-1}
%\end{equation}
%following the triangle inequality, the sub multiplicity property of the norm, and that $\Vert\bar{\boldsymbol{\Theta}}\Vert_2=1$.
by substituting \eqref{eq:sRI-sIT-bar-TC} and \eqref{eq:sRT-bar-TC} into \eqref{eq:UB-FC-1}, we have
\begin{equation}
\left\vert h^\star\right\vert^2
=\frac{1}{4Y_0^2}
\left(\left\vert y_{RT}-\frac{\bar{\mathbf{y}}_{RI}\bar{\mathbf{y}}_{IT}}{2Y_0}\right\vert+\frac{\left\Vert\bar{\mathbf{y}}_{RI}\right\Vert_2\left\Vert\bar{\mathbf{y}}_{IT}\right\Vert_2}{2Y_0}\right)^{2},\label{eq:UB-TC-2}
\end{equation}
and, by substituting \eqref{eq:yRI-bar} and \eqref{eq:yIT-bar} into \eqref{eq:UB-TC-2}, the maximum achievable channel gain can be rewritten as 
\begin{multline}
\left\vert h^\star\right\vert^2
=\frac{1}{4Y_0^2}
\left(\left\vert y_{RT}-\frac{1}{2}\mathbf{y}_{RI}\Re\{\mathbf{Y}_{II}\}^{-1}\mathbf{y}_{IT}\right\vert\right.\\
\left.+\frac{1}{2}\left\Vert\mathbf{y}_{RI}\Re\{\mathbf{Y}_{II}\}^{-1/2}\right\Vert_2\left\Vert\Re\{\mathbf{Y}_{II}\}^{-1/2}\mathbf{y}_{IT}\right\Vert_2\right)^{2},\label{eq:UB-TC-3}
\end{multline}
which is equal to the channel gain upper bound based on the $Z$-parameters given in \eqref{eq:UB-FC-3} as they both are tight upper bounds on the same objective function\footnote{It is possible to show that \eqref{eq:UB-TC-3} is equal to \eqref{eq:UB-FC-3} by substituting \eqref{eq:yRI-yIT}, \eqref{eq:yRT}, $Y_0=Z_0^{-1}$, and $\mathbf{Y}_{II}=\mathbf{Z}_{II}^{-1}$ into \eqref{eq:UB-TC-3}, but we omit the lengthy computations.}.
Remarkably, this analytical result establishes the optimality of tree-connected RIS also in the presence of mutual coupling.

%%%%%%%%%%%%%%%%%%%%%%%%%%%%%%%%%%%%%%%%%%%%%%%%%%
\section{Channel Gain Scaling Laws}
\label{sec:scaling-laws}

We have derived global optimal closed-form solutions to optimize fully- and tree-connected RISs in the presence of mutual coupling.
We have also provided the expression of the channel gain that can be exactly achieved through these solutions.
In this section, we provide the scaling law of the average channel gain achievable with mutual coupling and we compare it analytically with the average channel gain in the absence of mutual coupling.
To purely account for the effects of the RIS on the channel gain, we consider the direct link between transmitter and receiver to be completely obstructed, i.e., $z_{RT}=0$.
In addition, we assume the channels $\mathbf{z}_{RI}$ and $\mathbf{z}_{IT}$ to be independent and with independent Rayleigh distributed entries, i.e., $\mathbf{z}_{RI}\sim\mathcal{CN}(\mathbf{0},\rho_{RI}\mathbf{I})$ and $\mathbf{z}_{IT}\sim\mathcal{CN}(\mathbf{0},\rho_{IT}\mathbf{I})$, where $\rho_{RI}$ and $\rho_{IT}$ are the path gains\footnote{Despite the assumption of Rayleigh distributed channels in the theoretical derivations, it can be numerically shown that the obtained scaling laws are highly accurate also for independent Rician distributed channels, including Rayleigh and \gls{los} channels as special cases.}.
In the following, we consider the channel model based on the $Z$-parameters to derive the scaling law, while the same conclusions also hold for the $Y$-parameters model given the equivalence between two representations.

\subsection{Scaling Law with Mutual Coupling}

In the presence of mutual coupling and with obstructed direct link, the achievable channel gain is given by
\begin{multline}
\left\vert h_{MC}^\star\right\vert^2
=\frac{1}{16Z_0^2}
\left(\left\vert\mathbf{z}_{RI}\Re\{\mathbf{Z}_{II}\}^{-1}\mathbf{z}_{IT}\right\vert\right.\\
\left.+\left\Vert\mathbf{z}_{RI}\Re\{\mathbf{Z}_{II}\}^{-1/2}\right\Vert_2\left\Vert\Re\{\mathbf{Z}_{II}\}^{-1/2}\mathbf{z}_{IT}\right\Vert_2\right)^{2},\label{eq:UB-MC}
\end{multline}
which follows by substituting $z_{RT}=0$ in \eqref{eq:UB-FC-3}.
Thus, by taking the expectation of \eqref{eq:UB-MC}, we obtain
\begin{multline}
\mathbb{E}\left[\left\vert h_{MC}^\star\right\vert^2\right]
=\frac{1}{16Z_0^2}
\left(\mathbb{E}\left[\left\vert\mathbf{z}_{RI}\Re\{\mathbf{Z}_{II}\}^{-1}\mathbf{z}_{IT}\right\vert^2\right]\right.\\
+\mathbb{E}\left[\left\Vert\mathbf{z}_{RI}\Re\{\mathbf{Z}_{II}\}^{-1/2}\right\Vert_2^2\right]
 \mathbb{E}\left[\left\Vert\Re\{\mathbf{Z}_{II}\}^{-1/2}\mathbf{z}_{IT}\right\Vert_2^2\right]\\
+2\mathbb{E}\left[\left\vert\mathbf{z}_{RI}\Re\{\mathbf{Z}_{II}\}^{-1}\mathbf{z}_{IT}\right\vert\right]\\
\left.\times\mathbb{E}\left[\left\Vert\mathbf{z}_{RI}\Re\{\mathbf{Z}_{II}\}^{-1/2}\right\Vert_2\right]
            \mathbb{E}\left[\left\Vert\Re\{\mathbf{Z}_{II}\}^{-1/2}\mathbf{z}_{IT}\right\Vert_2\right]\right),\label{eq:E-UB-MC}
\end{multline}
where we exploited the independence between $\mathbf{z}_{RI}$ and $\mathbf{z}_{IT}$ and we considered the random variable $\vert\mathbf{z}_{RI}\Re\{\mathbf{Z}_{II}\}^{-1}\mathbf{z}_{IT}\vert$ to be approximately uncorrelated with $\Vert\mathbf{z}_{RI}\Re\{\mathbf{Z}_{II}\}^{-1/2}\Vert_2$ and $\Vert\Re\{\mathbf{Z}_{II}\}^{-1/2}\mathbf{z}_{IT}\Vert_2$.
The validity of this approximation will be verified in Section~\ref{sec:results}.
In the following, we provide a closed-form expression of $\mathbb{E}[\vert h_{MC}^\star\vert^2]$ by individually studying each expectation term in \eqref{eq:E-UB-MC}.

\textit{First}, the term $\mathbb{E}[\vert\mathbf{z}_{RI}\Re\{\mathbf{Z}_{II}\}^{-1}\mathbf{z}_{IT}\vert^2]$ in \eqref{eq:E-UB-MC} can be expressed as
\begin{align}
&\mathbb{E}\left[\left\vert\mathbf{z}_{RI}\Re\{\mathbf{Z}_{II}\}^{-1}\mathbf{z}_{IT}\right\vert^2\right]\\
&=\mathbb{E}\left[\mathbf{z}_{RI}\Re\{\mathbf{Z}_{II}\}^{-1}\mathbf{z}_{IT}\mathbf{z}_{IT}^H\Re\{\mathbf{Z}_{II}\}^{-1}\mathbf{z}_{RI}^H\right]\\
&=\mathbb{E}\left[\text{Tr}\left(\mathbf{z}_{RI}^H\mathbf{z}_{RI}\Re\{\mathbf{Z}_{II}\}^{-1}\mathbf{z}_{IT}\mathbf{z}_{IT}^H\Re\{\mathbf{Z}_{II}\}^{-1}\right)\right]\\
&=\text{Tr}\left(\mathbb{E}\left[\mathbf{z}_{RI}^H\mathbf{z}_{RI}\right]\Re\{\mathbf{Z}_{II}\}^{-1}\mathbb{E}\left[\mathbf{z}_{IT}\mathbf{z}_{IT}^H\right]\Re\{\mathbf{Z}_{II}\}^{-1}\right),\label{eq:E-UB-MC-1-tmp}
\end{align}
by exploiting the symmetry of the Frobenius inner product, the linearity of the trace, and that $\mathbf{z}_{RI}$ and $\mathbf{z}_{IT}$ are independent.
By also noticing that $\mathbb{E}[\mathbf{z}_{RI}^H\mathbf{z}_{RI}]=\rho_{RI}\mathbf{I}$ and $\mathbb{E}[\mathbf{z}_{IT}\mathbf{z}_{IT}^H]=\rho_{IT}\mathbf{I}$ since $\mathbf{z}_{RI}\sim\mathcal{CN}(\mathbf{0},\rho_{RI}\mathbf{I})$ and $\mathbf{z}_{IT}\sim\mathcal{CN}(\mathbf{0},\rho_{IT}\mathbf{I})$, \eqref{eq:E-UB-MC-1-tmp} becomes
\begin{equation}
\mathbb{E}\left[\left\vert\mathbf{z}_{RI}\Re\{\mathbf{Z}_{II}\}^{-1}\mathbf{z}_{IT}\right\vert^2\right]
=\rho_{RI}\rho_{IT}\text{Tr}\left(\Re\{\mathbf{Z}_{II}\}^{-2}\right).\label{eq:E-UB-MC-1}
\end{equation}

\textit{Second}, $\mathbb{E}[\Vert\mathbf{z}_{RI}\Re\{\mathbf{Z}_{II}\}^{-1/2}\Vert_2^2]$ in \eqref{eq:E-UB-MC} can be rewritten as
\begin{align}
&\mathbb{E}\left[\left\Vert\mathbf{z}_{RI}\Re\{\mathbf{Z}_{II}\}^{-1/2}\right\Vert_2^2\right]\\
&=\mathbb{E}\left[\mathbf{z}_{RI}\Re\{\mathbf{Z}_{II}\}^{-1}\mathbf{z}_{RI}^H\right]\\
&=\mathbb{E}\left[\text{Tr}\left(\Re\{\mathbf{Z}_{II}\}^{-1}\mathbf{z}_{RI}^H\mathbf{z}_{RI}\right)\right]\\
&=\text{Tr}\left(\Re\{\mathbf{Z}_{II}\}^{-1}\mathbb{E}\left[\mathbf{z}_{RI}^H\mathbf{z}_{RI}\right]\right),
\end{align}
following the symmetry of the Frobenius inner product and the linearity of the trace.
Thus, by recalling that $\mathbb{E}[\mathbf{z}_{RI}^H\mathbf{z}_{RI}]=\rho_{RI}\mathbf{I}$, we obtain
\begin{equation}
\mathbb{E}\left[\left\Vert\mathbf{z}_{RI}\Re\{\mathbf{Z}_{II}\}^{-1/2}\right\Vert_2^2\right]
=\rho_{RI}\text{Tr}\left(\Re\{\mathbf{Z}_{II}\}^{-1}\right).\label{eq:E-UB-MC-2a}
\end{equation}
With a similar argument, it is possible to show that
\begin{equation}
\mathbb{E}\left[\left\Vert\Re\{\mathbf{Z}_{II}\}^{-1/2}\mathbf{z}_{IT}\right\Vert_2^2\right]
=\rho_{IT}\text{Tr}\left(\Re\{\mathbf{Z}_{II}\}^{-1}\right),\label{eq:E-UB-MC-2b}
\end{equation}
since $\mathbb{E}[\mathbf{z}_{IT}\mathbf{z}_{IT}^H]=\rho_{IT}\mathbf{I}$.

%https://math.stackexchange.com/questions/767835/proving-eigenvalue-squared-is-eigenvalue-of-a2
\textit{Third}, to compute $\mathbb{E}[\vert\mathbf{z}_{RI}\Re\{\mathbf{Z}_{II}\}^{-1}\mathbf{z}_{IT}\vert]$ in \eqref{eq:E-UB-MC}, we use the eigenvalue decomposition of $\Re\{\mathbf{Z}_{II}\}^{-1}$.
Since $\Re\{\mathbf{Z}_{II}\}^{-1}$ is a real symmetric matrix, we have $\Re\{\mathbf{Z}_{II}\}^{-1}=\mathbf{Q}\mathbf{\Lambda}\mathbf{Q}^T$, where $\mathbf{Q}\in\mathbb{R}^{N_I\times N_I}$ is an orthogonal matrix whose columns are the eigenvectors of $\Re\{\mathbf{Z}_{II}\}^{-1}$, and $\mathbf{\Lambda}=\text{diag}(\lambda_1,\ldots,\lambda_{N_I})$, with $\lambda_{n_I}\in\mathbb{R}$ being the $n_I$th eigenvalue of $\Re\{\mathbf{Z}_{II}\}^{-1}$.
Consequently, the random variable $R=\mathbf{z}_{RI}\Re\{\mathbf{Z}_{II}\}^{-1}\mathbf{z}_{IT}$ can be rewritten as $R=\mathbf{q}_{RI}\mathbf{\Lambda}\mathbf{q}_{IT}$, where $\mathbf{q}_{RI}=\mathbf{z}_{RI}\mathbf{Q}$ is distributed as $\mathbf{q}_{RI}\sim\mathcal{CN}(\mathbf{0},\rho_{RI}\mathbf{I})$ and $\mathbf{q}_{IT}=\mathbf{Q}^T\mathbf{z}_{IT}$ is distributed as $\mathbf{q}_{IT}\sim\mathcal{CN}(\mathbf{0},\rho_{IT}\mathbf{I})$ due to the orthogonality of $\mathbf{Q}$.
Since $R$ can be expressed as $R=\mathbf{q}_{RI}\mathbf{\Lambda}\mathbf{q}_{IT}$, $R$ is the sum of $N_I$ independent random variables, i.e., $R=\sum_{n_I=1}^{N_I}R_{n_I}$, where $R_{n_I}=[\mathbf{q}_{RI}]_{n_I}\lambda_{n_I}[\mathbf{q}_{IT}]_{n_I}$ has variance $\rho_{RI}\rho_{IT}\lambda_{n_I}^2$, for $n_I=1,\ldots,N_I$.
Thus, according to the Lyapunov \gls{clt}, $R$ is distributed as $R\sim\mathcal{CN}(0,\rho_{RI}\rho_{IT}\sum_{n_I=1}^{N_I}\lambda_{n_I}^2)$.
Since $\lambda_{n_I}^2$, for $n_I=1,\ldots,N_I$, are the eigenvalues of $\Re\{\mathbf{Z}_{II}\}^{-2}$, we have $\sum_{n_I=1}^{N_I}\lambda_{n_I}^2=\text{Tr}(\Re\{\mathbf{Z}_{II}\}^{-2})$, giving $R\sim\mathcal{CN}(0,\rho_{RI}\rho_{IT}\text{Tr}\left(\Re\{\mathbf{Z}_{II}\}^{-2}\right))$ and
\begin{equation}
\mathbb{E}\left[\left\vert\mathbf{z}_{RI}\Re\{\mathbf{Z}_{II}\}^{-1}\mathbf{z}_{IT}\right\vert\right]
=\sqrt{\frac{\pi}{4}\rho_{RI}\rho_{IT}\text{Tr}\left(\Re\{\mathbf{Z}_{II}\}^{-2}\right)},\label{eq:E-UB-MC-3}
\end{equation}
following the expression of the mean of the Rayleigh distribution.

\textit{Fourth}, to obtain $\mathbb{E}[\Vert\mathbf{z}_{RI}\Re\{\mathbf{Z}_{II}\}^{-1/2}\Vert_2]$ in \eqref{eq:E-UB-MC}, we exploit the fact that the squared norm $\Vert\mathbf{z}_{RI}\Re\{\mathbf{Z}_{II}\}^{-1/2}\Vert_2^2$ becomes deterministic as the number of RIS elements $N_I$ increases because of the so-called channel hardening phenomena \cite{ngo17}, i.e., 
\begin{equation}
\text{Var}\left(\frac{\left\Vert\mathbf{z}_{RI}\Re\{\mathbf{Z}_{II}\}^{-1/2}\right\Vert_2^2}{\mathbb{E}\left[\left\Vert\mathbf{z}_{RI}\Re\{\mathbf{Z}_{II}\}^{-1/2}\right\Vert_2^2\right]}\right)\rightarrow0,\label{eq:channel-hardening}
\end{equation}
as $N_I\rightarrow\infty$.
To show that \eqref{eq:channel-hardening} holds as $N_I\rightarrow\infty$, we consider the eigenvalue decomposition $\Re\{\mathbf{Z}_{II}\}^{-1}=\mathbf{Q}\mathbf{\Lambda}\mathbf{Q}^T$ to rewrite the random variable $S=\Vert\mathbf{z}_{RI}\Re\{\mathbf{Z}_{II}\}^{-1/2}\Vert_2^2$ as $S=\mathbf{q}_{RI}\mathbf{\Lambda}\mathbf{q}_{RI}^H$, where $\mathbf{q}_{RI}=\mathbf{z}_{RI}\mathbf{Q}\sim\mathcal{CN}(\mathbf{0},\rho_{RI}\mathbf{I})$.
%https://math.stackexchange.com/questions/620045/mean-and-variance-of-squared-gaussian-y-x2-where-x-sim-mathcaln0-sigma
Thus, $S$ is the sum of $N_I$ independent random variables, i.e., $S=\sum_{n_I=1}^{N_I}S_{n_I}$, where $S_{n_I}=\vert[\mathbf{q}_{RI}]_{n_I}\vert^2\lambda_{n_I}$ has variance $\rho_{RI}^2\lambda_{n_I}^2$, for $n_I=1,\ldots,N_I$, and, following the Lyapunov \gls{clt}, we have
\begin{equation}
\text{Var}\left(\left\Vert\mathbf{z}_{RI}\Re\{\mathbf{Z}_{II}\}^{-1/2}\right\Vert_2^2\right)=\rho_{RI}^2\text{Tr}\left(\Re\{\mathbf{Z}_{II}\}^{-2}\right),\label{eq:channel-hardening-1}
\end{equation}
since $\sum_{n_I=1}^{N_I}\lambda_{n_I}^2=\text{Tr}(\Re\{\mathbf{Z}_{II}\}^{-2})$.
By using \eqref{eq:channel-hardening-1} and observing that $\mathbb{E}[\Vert\mathbf{z}_{RI}\Re\{\mathbf{Z}_{II}\}^{-1/2}\Vert_2^2]^2=\rho_{RI}^2\text{Tr}(\Re\{\mathbf{Z}_{II}\}^{-1})^2$ because of \eqref{eq:E-UB-MC-2a}, the variance of the ratio in \eqref{eq:channel-hardening} can be rewritten as
\begin{equation}
\text{Var}\left(\frac{\left\Vert\mathbf{z}_{RI}\Re\{\mathbf{Z}_{II}\}^{-1/2}\right\Vert_2^2}{\mathbb{E}\left[\left\Vert\mathbf{z}_{RI}\Re\{\mathbf{Z}_{II}\}^{-1/2}\right\Vert_2^2\right]}\right)
=\frac{\text{Tr}\left(\Re\{\mathbf{Z}_{II}\}^{-2}\right)}{\text{Tr}\left(\Re\{\mathbf{Z}_{II}\}^{-1}\right)^2}.
\end{equation}
By considering the inequalities $\text{Tr}\left(\Re\{\mathbf{Z}_{II}\}^{-2}\right)\leq\lambda_1^2N_I$, where $\lambda_1$ is the dominant eigenvalue of $\Re\{\mathbf{Z}_{II}\}^{-1}$, and $\text{Tr}\left(\Re\{\mathbf{Z}_{II}\}^{-1}\right)\geq\lambda_{N_I}N_I$, where $\lambda_{N_I}$ is the smallest eigenvalue of $\Re\{\mathbf{Z}_{II}\}^{-1}$, we can write
\begin{align}
\text{Var}\left(\frac{\left\Vert\mathbf{z}_{RI}\Re\{\mathbf{Z}_{II}\}^{-1/2}\right\Vert_2^2}{\mathbb{E}\left[\left\Vert\mathbf{z}_{RI}\Re\{\mathbf{Z}_{II}\}^{-1/2}\right\Vert_2^2\right]}\right)
&\leq\frac{\lambda_1^2N_I}{\lambda_{N_I}^2N_I^2}\\
&=\frac{k^2}{N_I}\rightarrow0,
\end{align}
as $N_I\rightarrow\infty$, where we introduced the condition number of $\Re\{\mathbf{Z}_{II}\}^{-1}$ as $k=\lambda_1/\lambda_{N_I}$.

Since $\Vert\mathbf{z}_{RI}\Re\{\mathbf{Z}_{II}\}^{-1/2}\Vert_2^2$ fluctuates only a little around its mean value, it can be approximated as deterministic, i.e., 
\begin{align}
\left\Vert\mathbf{z}_{RI}\Re\{\mathbf{Z}_{II}\}^{-1/2}\right\Vert_2^2
\overset{N_I\nearrow}{\approx}&\mathbb{E}\left[\left\Vert\mathbf{z}_{RI}\Re\{\mathbf{Z}_{II}\}^{-1/2}\right\Vert_2^2\right]\\
=&\rho_{RI}\text{Tr}\left(\Re\{\mathbf{Z}_{II}\}^{-1}\right),
\end{align}
as $N_I\rightarrow\infty$, following \eqref{eq:E-UB-MC-2a}.
Consequently, we have
\begin{equation}
\mathbb{E}\left[\left\Vert\mathbf{z}_{RI}\Re\{\mathbf{Z}_{II}\}^{-1/2}\right\Vert_2\right]
\overset{N_I\nearrow}{\approx}\sqrt{\rho_{RI}\text{Tr}\left(\Re\{\mathbf{Z}_{II}\}^{-1}\right)}.\label{eq:E-UB-MC-4a}
\end{equation}
In a similar way, it can also be shown that 
\begin{equation}
\mathbb{E}\left[\left\Vert\Re\{\mathbf{Z}_{II}\}^{-1/2}\mathbf{z}_{IT}\right\Vert_2\right]
\overset{N_I\nearrow}{\approx}\sqrt{\rho_{IT}\text{Tr}\left(\Re\{\mathbf{Z}_{II}\}^{-1}\right)}.\label{eq:E-UB-MC-4b}
\end{equation}
Interestingly, although the expressions in \eqref{eq:E-UB-MC-4a} and \eqref{eq:E-UB-MC-4b} are obtained by exploiting the channel hardening approximation valid for $N_I\rightarrow\infty$, they remain precise even for practical numbers of RIS elements, such as $N_I=64$, as will be illustrated in Section~\ref{sec:results}.

Given the expressions of all the expectation terms in \eqref{eq:E-UB-MC}, the expression of $\mathbb{E}[\vert h_{MC}^\star\vert^2]$ can now be obtained.
Specifically, the scaling law of the average channel gain with mutual coupling under independent Rayleigh fading channels is finally given by
\begin{multline}
\mathbb{E}\left[\left\vert h_{MC}^\star\right\vert^2\right]
=\frac{\rho_{RI}\rho_{IT}}{16Z_0^2}
\left(\text{Tr}\left(\Re\{\mathbf{Z}_{II}\}^{-2}\right)
+\text{Tr}\left(\Re\{\mathbf{Z}_{II}\}^{-1}\right)^2\right.\\
\left.+\sqrt{\pi\text{Tr}\left(\Re\{\mathbf{Z}_{II}\}^{-2}\right)}
\text{Tr}\left(\Re\{\mathbf{Z}_{II}\}^{-1}\right)\right),\label{eq:E-UB-MC-CF}
\end{multline}
which is obtained by substituting \eqref{eq:E-UB-MC-1}, \eqref{eq:E-UB-MC-2a}, \eqref{eq:E-UB-MC-2b}, \eqref{eq:E-UB-MC-3}, \eqref{eq:E-UB-MC-4a}, and \eqref{eq:E-UB-MC-4b} in \eqref{eq:E-UB-MC}.
Interestingly, $\mathbb{E}[\vert h_{MC}^\star\vert^2]$ solely depends on the product of the channel gains of $\mathbf{z}_{RI}$ and $\mathbf{z}_{IT}$ denoted as $\rho_{RI}\rho_{IT}$, and on the two traces $\text{Tr}(\Re\{\mathbf{Z}_{II}\}^{-2})$ and $\text{Tr}(\Re\{\mathbf{Z}_{II}\}^{-1})$.

\begin{table*}[t]
\centering
\caption{Expressions of the channel gains and their scaling laws with and without mutual coupling.}
\begin{tabular}{@{}ll@{}}
\toprule
 & Expression\\
\midrule
Channel gain with mutual coupling &
$\left\vert h_{MC}^\star\right\vert^2
=\frac{1}{16Z_0^2}
\left(\left\vert\mathbf{z}_{RI}\Re\{\mathbf{Z}_{II}\}^{-1}\mathbf{z}_{IT}\right\vert
+\left\Vert\mathbf{z}_{RI}\Re\{\mathbf{Z}_{II}\}^{-1/2}\right\Vert_2\left\Vert\Re\{\mathbf{Z}_{II}\}^{-1/2}\mathbf{z}_{IT}\right\Vert_2\right)^{2}$
\\
Channel gain with no mutual coupling &
$\left\vert h_{NoMC}^\star\right\vert^2
=\frac{1}{16Z_0^2Z_{II}^2}
\left(\left\vert\mathbf{z}_{RI}\mathbf{z}_{IT}\right\vert
+\left\Vert\mathbf{z}_{RI}\right\Vert_2\left\Vert\mathbf{z}_{IT}\right\Vert_2\right)^{2}$
\\
Scaling law with mutual coupling &
$\mathbb{E}\left[\left\vert h_{MC}^\star\right\vert^2\right]
=\frac{\rho_{RI}\rho_{IT}}{16Z_0^2}
\left(\text{Tr}\left(\Re\{\mathbf{Z}_{II}\}^{-2}\right)
+\text{Tr}\left(\Re\{\mathbf{Z}_{II}\}^{-1}\right)^2
+\sqrt{\pi\text{Tr}\left(\Re\{\mathbf{Z}_{II}\}^{-2}\right)}
\text{Tr}\left(\Re\{\mathbf{Z}_{II}\}^{-1}\right)\right)$
\\
Scaling law with no mutual coupling &
$\mathbb{E}\left[\left\vert h_{NoMC}^\star\right\vert^2\right]
=\frac{\rho_{RI}\rho_{IT}}{16Z_0^2Z_{II}^2}
\left(N_I+N_I^2+\sqrt{\pi N_I}N_I\right)$
\\
\bottomrule
\end{tabular}
\label{tab:exp}
\end{table*}

\subsection{Scaling Law with No Mutual Coupling}

In the absence of mutual coupling, but with the possible presence of impedance mismatching, the mutual coupling matrix $\mathbf{Z}_{II}$ is diagonal, where its diagonal entries represent the self-impedance of the RIS antennas.
Thus, in this case, we have $\Re\{\mathbf{Z}_{II}\}=Z_{II}\mathbf{I}$, where $Z_{II}\in\mathbb{R}$ is the real part of the self-impedance of the RIS antennas, being $Z_{II}=Z_0$ in the case of perfect matching.
By substituting $\Re\{\mathbf{Z}_{II}\}=Z_{II}\mathbf{I}$ in \eqref{eq:UB-MC}, we obtain the achievable channel gain in the absence of mutual coupling as
\begin{equation}
\left\vert h_{NoMC}^\star\right\vert^2
=\frac{1}{16Z_0^2Z_{II}^2}
\left(\left\vert\mathbf{z}_{RI}\mathbf{z}_{IT}\right\vert
+\left\Vert\mathbf{z}_{RI}\right\Vert_2\left\Vert\mathbf{z}_{IT}\right\Vert_2\right)^{2}.\label{eq:UB-NoMC}
\end{equation}
By taking the expectation, $\mathbb{E}[\vert h_{NoMC}^\star\vert^2]$ writes as
\begin{multline}
\mathbb{E}\left[\left\vert h_{NoMC}^\star\right\vert^2\right]
=\frac{1}{16Z_0^2Z_{II}^2}
\left(\mathbb{E}\left[\left\vert\mathbf{z}_{RI}\mathbf{z}_{IT}\right\vert^2\right]\right.\\
+\mathbb{E}\left[\left\Vert\mathbf{z}_{RI}\right\Vert_2^2\right]
 \mathbb{E}\left[\left\Vert\mathbf{z}_{IT}\right\Vert_2^2\right]\\
\left.+2\mathbb{E}\left[\left\vert\mathbf{z}_{RI}\mathbf{z}_{IT}\right\vert\right]
        \mathbb{E}\left[\left\Vert\mathbf{z}_{RI}\right\Vert_2\right]
        \mathbb{E}\left[\left\Vert\mathbf{z}_{IT}\right\Vert_2\right]\right),\label{eq:E-UB-NoMC}
\end{multline}
where we exploited the independence between $\mathbf{z}_{RI}$ and $\mathbf{z}_{IT}$ and appoximated the random variable $\vert\mathbf{z}_{RI}\mathbf{z}_{IT}\vert$ as uncorrelated with $\Vert\mathbf{z}_{RI}\Vert_2$ and $\Vert\mathbf{z}_{IT}\Vert_2$.
This approximation will be validated in Section~\ref{sec:results}.
In the following, we derive a closed-form expression for \eqref{eq:E-UB-NoMC} by individually studying all the expectation terms therein.

\textit{First}, the term $\mathbb{E}[\vert\mathbf{z}_{RI}\mathbf{z}_{IT}\vert^2]$ in \eqref{eq:E-UB-NoMC} can be rewritten as
\begin{align}
\mathbb{E}\left[\left\vert\mathbf{z}_{RI}\mathbf{z}_{IT}\right\vert^2\right]
&=\mathbb{E}\left[\mathbf{z}_{RI}\mathbf{z}_{IT}\mathbf{z}_{IT}^H\mathbf{z}_{RI}^H\right]\\
&=\mathbb{E}\left[\text{Tr}\left(\mathbf{z}_{RI}^H\mathbf{z}_{RI}\mathbf{z}_{IT}\mathbf{z}_{IT}^H\right)\right]\\
&=\text{Tr}\left(\mathbb{E}\left[\mathbf{z}_{RI}^H\mathbf{z}_{RI}\right]\mathbb{E}\left[\mathbf{z}_{IT}\mathbf{z}_{IT}^H\right]\right),
\end{align}
which simplifies as 
\begin{equation}
\mathbb{E}\left[\left\vert\mathbf{z}_{RI}\mathbf{z}_{IT}\right\vert^2\right]
=\rho_{RI}\rho_{IT}N_I,\label{eq:E-UB-NoMC-1}
\end{equation}
since $\mathbb{E}[\mathbf{z}_{RI}^H\mathbf{z}_{RI}]=\rho_{RI}\mathbf{I}$ and $\mathbb{E}[\mathbf{z}_{IT}\mathbf{z}_{IT}^H]=\rho_{IT}\mathbf{I}$.

\textit{Second}, $\mathbb{E}[\Vert\mathbf{z}_{RI}\Vert_2^2]$ and $\mathbb{E}[\Vert\mathbf{z}_{IT}\Vert_2^2]$ in \eqref{eq:E-UB-NoMC} are given by
\begin{equation}
\mathbb{E}\left[\left\Vert\mathbf{z}_{RI}\right\Vert_2^2\right]
=\rho_{RI}N_I,\;
\mathbb{E}\left[\left\Vert\mathbf{z}_{IT}\right\Vert_2^2\right]
=\rho_{IT}N_I,\label{eq:E-UB-NoMC-2}
\end{equation}
by exploiting the second moment of the chi distribution with $2N_I$ degrees of freedom.

\textit{Third}, $\mathbb{E}\left[\left\vert\mathbf{z}_{RI}\mathbf{z}_{IT}\right\vert\right]$ in \eqref{eq:E-UB-NoMC} can be computed by noticing that the random variable $T=\mathbf{z}_{RI}\mathbf{z}_{IT}$ is the sum of $N_I$ independent random variables, each with variance $\rho_{RI}\rho_{IT}$.
Thus, according to the \gls{clt}, $T$ is distributed as $T\sim\mathcal{CN}(0,\rho_{RI}\rho_{IT}N_I)$.
By recalling the expression of the mean of the Rayleigh distribution, we eventually obtain
\begin{equation}
\mathbb{E}\left[\left\vert\mathbf{z}_{RI}\mathbf{z}_{IT}\right\vert\right]
=\sqrt{\frac{\pi}{4}\rho_{RI}\rho_{IT}N_I}.\label{eq:E-UB-NoMC-3}
\end{equation}

\textit{Fourth}, $\mathbb{E}[\Vert\mathbf{z}_{RI}\Vert_2]$ and $\mathbb{E}[\Vert\mathbf{z}_{IT}\Vert_2]$ in \eqref{eq:E-UB-NoMC} are given by
\begin{align}
\mathbb{E}\left[\left\Vert\mathbf{z}_{RI}\right\Vert_2\right]
&=\sqrt{\rho_{RI}}\frac{\Gamma\left(N_I+1/2\right)}{\Gamma\left(N_I\right)},\label{eq:E-UB-NoMC-4a-tmp}\\
\mathbb{E}\left[\left\Vert\mathbf{z}_{IT}\right\Vert_2\right]
&=\sqrt{\rho_{IT}}\frac{\Gamma\left(N_I+1/2\right)}{\Gamma\left(N_I\right)},\label{eq:E-UB-NoMC-4b-tmp}
\end{align}
respectively, where $\Gamma(\cdot)$ is the Gamma function, following the first moment of the chi distribution with $2N_I$ degrees of freedom.
The expressions of $\mathbb{E}[\Vert\mathbf{z}_{RI}\Vert_2]$ and $\mathbb{E}[\Vert\mathbf{z}_{IT}\Vert_2]$ in \eqref{eq:E-UB-NoMC-4a-tmp} and \eqref{eq:E-UB-NoMC-4b-tmp} can be simplified by using the Laurent series expansion of $\Gamma(N_I+1/2)/\Gamma(N_I)$ at $N_I\rightarrow\infty$, i.e.,
\begin{equation}
\frac{\Gamma\left(N_I+1/2\right)}{\Gamma\left(N_I\right)}=\sqrt{N_I-\frac{1}{4}+\mathcal{O}\left(\frac{1}{N_I}\right)},\label{eq:Laurent}
\end{equation}
as $N_I\rightarrow\infty$ \cite{ahl53}.
From \eqref{eq:Laurent}, we notice that the function $\Gamma(N_I+1/2)/\Gamma(N_I)$ is well approximated by $\sqrt{N_I}$ for high values of $N_I$, allowing to rewrite \eqref{eq:E-UB-NoMC-4a-tmp} and \eqref{eq:E-UB-NoMC-4b-tmp} as
\begin{equation}
\mathbb{E}\left[\left\Vert\mathbf{z}_{RI}\right\Vert_2\right]
=\sqrt{\rho_{RI}N_I},\;
\mathbb{E}\left[\left\Vert\mathbf{z}_{IT}\right\Vert_2\right]
=\sqrt{\rho_{IT}N_I}.\label{eq:E-UB-NoMC-4}
\end{equation}
Remarkably, despite the approximation in \eqref{eq:E-UB-NoMC-4} is derived from the Laurent series expansion at $N_I\rightarrow\infty$, it is highly accurate also for practical numbers of RIS elements, as it will be shown in Section~\ref{sec:results}.

By substituting \eqref{eq:E-UB-NoMC-1}, \eqref{eq:E-UB-NoMC-2}, \eqref{eq:E-UB-NoMC-3}, and \eqref{eq:E-UB-NoMC-4} in \eqref{eq:E-UB-NoMC}, we finally obtain the scaling law of the average channel gain with no mutual coupling as
\begin{equation}
\mathbb{E}\left[\left\vert h_{NoMC}^\star\right\vert^2\right]
=\frac{\rho_{RI}\rho_{IT}}{16Z_0^2Z_{II}^2}
\left(N_I+N_I^2+\sqrt{\pi N_I}N_I\right).\label{eq:E-UB-NoMC-CF}
\end{equation}
Note that \eqref{eq:E-UB-NoMC-CF} depends solely on the product of the channel gains $\rho_{RI}\rho_{IT}$, on the real part of the RIS antennas self-impedance $Z_{II}$, and on the number of RIS elements $N_I$.
Furthermore, \eqref{eq:E-UB-NoMC-CF} can be seen as a special case of \eqref{eq:E-UB-MC-CF} in which $\Re\{\mathbf{Z}_{II}\}=Z_{II}\mathbf{I}$, although \eqref{eq:E-UB-NoMC-CF} and \eqref{eq:E-UB-MC-CF} have been derived with different strategies and approximations.
The expression in \eqref{eq:E-UB-NoMC-CF}, proportional to $N_I^2+\sqrt{\pi N_I}N_I+N_I$, differs from the scaling law of fully-connected RISs with no mutual coupling provided in \cite{she20}, i.e., $N_I^2$, since \cite{she20} neglects the effects of the structural scattering, or specular reflection, of the RIS, as clarified in \cite[Section~V]{ner23-2}.
In Tab.~\ref{tab:exp}, we summarize the findings of this section by reporting the expressions of the channel gains with and without mutual coupling, together with their corresponding scaling laws.

%%%%%%%%%%%%%%%%%%%%%%%%%%%%%%%%%%%%%%%%%%%%%%%%%%
\section{Impact of Mutual Coupling\\on the Average Channel Gain}
\label{sec:impact}

We have derived the scaling law of the average channel gain in the presence and absence of mutual coupling in \eqref{eq:E-UB-MC-CF} and \eqref{eq:E-UB-NoMC-CF}, respectively.
In this section, we analytically assess whether mutual coupling is detrimental or beneficial in enhancing the average channel gain.
To compare the average channel gain in the presence and in the absence of mutual coupling, we recall that the scaling law in \eqref{eq:E-UB-MC-CF} is determined by two trace terms, i.e., $\text{Tr}(\Re\{\mathbf{Z}_{II}\}^{-2})$ and $\mathrm{Tr}(\Re\{\mathbf{Z}_{II}\}^{-1})$, and introduce the following two lemmas to gain insights into the two trace terms.

\begin{lemma}
Given a positive definite matrix $\mathbf{A}\in\mathbb{R}^{N\times N}$ with diagonal elements $[\mathbf{A}]_{n,n}=a$, for $n=1,\ldots,N$, it holds
\begin{equation}
\emph{Tr}\left(\mathbf{A}^{-1}\right)\geq\frac{N}{a}.\label{eq:lem1}
\end{equation}
\label{lem:1}
\end{lemma}
\begin{proof}Please refer to Appendix~B.\end{proof}

\begin{lemma}
Given a positive definite matrix $\mathbf{A}\in\mathbb{R}^{N\times N}$ with diagonal elements $[\mathbf{A}]_{n,n}=a$, for $n=1,\ldots,N$, it holds
\begin{equation}
\emph{Tr}\left(\mathbf{A}^{-2}\right)\geq\frac{N}{a^2}.\label{eq:lem2}
\end{equation}
\label{lem:2}
\end{lemma}
\begin{proof}Please refer to Appendix~C.\end{proof}

By using the results in Lemmas~\ref{lem:1} and \ref{lem:2}, we can derive the following proposition, stating that mutual coupling is beneficial as it improves the average channel gain.

\begin{proposition}
Under independent Rayleigh fading channels, mutual coupling between the RIS elements improves the average channel gain, i.e., 
\begin{equation}
\mathbb{E}\left[\left\vert h_{MC}^\star\right\vert^2\right]\geq\mathbb{E}\left[\left\vert h_{NoMC}^\star\right\vert^2\right],
\end{equation}
where $\mathbb{E}[\vert h_{MC}^\star\vert^2]$ and $\mathbb{E}[\vert h_{NoMC}^\star\vert^2]$ are given by \eqref{eq:E-UB-MC-CF} and \eqref{eq:E-UB-NoMC-CF}.
\label{pro:impact}
\end{proposition}
\begin{proof}Please refer to Appendix~D.\end{proof}

According to Proposition~\ref{pro:impact}, mutual coupling is always beneficial in enhancing the average channel gain.
Note that this result holds for any mutual coupling matrix $\mathbf{Z}_{II}$ of a RIS whose elements have all the same self-impedance $Z_{II}$, which can also be $Z_{II}\neq Z_0$ in the case of impedance mismatching.
Thus, it is valid for any geometry of the RIS antenna array as long as all RIS antenna elements are of the same type, e.g., dipole or patch antenna, and also in the presence of impedance mismatch.

% Relation with previous research
To relate Proposition~\ref{pro:impact} to previous research on RIS with mutual coupling, we make the following two remarks.
\textit{First}, it is worth emphasizing that the presence of mutual coupling enhances the channel gain, on average.
However, mutual coupling can be beneficial for certain channel realizations and detrimental for others.
In light of this, our result does not conflict with previous works where mutual coupling was found to be detrimental for a specific channel realization \cite{qia21,li24}.
\textit{Second}, Proposition~\ref{pro:impact} states that mutual coupling is beneficial when the RIS is a fully- or tree-connected RIS, and its mutual coupling is accounted for during the RIS optimization process.
Thus, it does not give any insight into the effect of mutual coupling on conventional D-RIS, whose analytical characterization remains an open problem up to our best knowledge.
In addition, the presence of mutual coupling could reduce the channel gain if mutual coupling is not considered in the RIS optimization process.
The performance of conventional D-RIS and the performance degradation due to a mutual coupling-unaware optimization are numerically investigated in Section~\ref{sec:results}.

%%%%%%%%%%%%%%%%%%%%%%%%%%%%%%%%%%%%%%%%%%%%%%%%%%
\section{Numerical Results}
\label{sec:results}

% Setup
This section provides simulation results to validate the derived global optimal closed-form solutions for fully- and tree-connected RISs and the channel gain scaling laws.
We consider a RIS implemented as a uniform planar array (UPA) of radiating elements located in the $x$-$y$ plane, with dimensions $N_x\times N_y$, where $N_x=8$ and $N_y=N_I/8$, and with inter-element distance $d$.
The RIS elements are thin wire dipoles parallel to the $y$ axis with length $\ell=\lambda/4$ and radius $r\ll\ell$, where $\lambda=c/f$ is the wavelength of the frequency $f=28$~GHz, and $c$ is the speed of light.
All the RIS elements are assumed to be perfectly matched to $Z_0=50\;\Omega$, giving $\left[\mathbf{Z}_{II}\right]_{n_I,n_I}=Z_0$, for $n_I=1,\ldots,N_I$.
Besides, the $(q,p)$th entry of $\mathbf{Z}_{II}$, with $q\neq p$, represent the mutual coupling between the RIS element $p$ located in $(x_p,y_p)$ and the RIS element $q$ located in $(x_q,y_q)$.
Following \cite{gra21,li24}, $\left[\mathbf{Z}_{II}\right]_{q,p}=\left[\mathbf{Z}_{II}\right]_{p,q}$ is modeled as
\begin{multline}
\left[\mathbf{Z}_{II}\right]_{q,p}=
\int_{y_q-\frac{\ell}{2}}^{y_q+\frac{\ell}{2}}
\int_{y_p-\frac{\ell}{2}}^{y_p+\frac{\ell}{2}}
\frac{j\eta_0}{4\pi k_0}
\left(\frac{\left(y''-y'\right)^2}{d_{q,p}^2}\right.\\
\left.\times\left(\frac{3}{d_{q,p}^2}+\frac{3jk_0}{d_{q,p}}-k_0^2\right)
-\frac{jk_0+d_{q,p}^{-1}}{d_{q,p}}+k_0^2\right)
\frac{e^{-jk_0d_{q,p}}}{d_{q,p}}\\
\times\frac{
\sin\left(k_0\left(\frac{\ell}{2}-\left\vert y'-y_p\right\vert\right)\right)
\sin\left(k_0\left(\frac{\ell}{2}-\left\vert y''-y_q\right\vert\right)\right)}
{\sin^2\left(k_0\frac{\ell}{2}\right)}dy'dy'',
\end{multline}
where $\eta_0=377\;\Omega$ is the impedance of free space, $k_0=2\pi/\lambda$ is the wavenumber, and $d_{q,p}=\sqrt{(x_q-x_p)^2+(y''-y')^2}$.
We consider the direct link between the transmitter and receiver to be fully obstructed, i.e., $z_{RT}=0$, and generate $\mathbf{z}_{RI}$ and $\mathbf{z}_{IT}$ as independent Rayleigh distributed, i.e., $\mathbf{z}_{RI}\sim\mathcal{CN}(\mathbf{0},\rho_{RI}\mathbf{I})$ and $\mathbf{z}_{IT}\sim\mathcal{CN}(\mathbf{0},\rho_{IT}\mathbf{I})$, with path gains $\rho_{RI}=\rho_{IT}=4Z_0^{2}10^{-8}$.
Although the path gains $\rho_{RI}$ and $\rho_{IT}$ depend on the transmission distances, we directly set the value of $\rho_{RI}$ and $\rho_{IT}$ since they impact the channel gain by just scaling it by the product $\rho_{RI}\rho_{IT}$.
%$\mathbf{z}_{RI}/(2Z_0)\sim\mathcal{CN}(\mathbf{0},G_{RI}\mathbf{I})$ and $\mathbf{z}_{IT}/(2Z_0)\sim\mathcal{CN}(\mathbf{0},G_{IT}\mathbf{I})$, with path gain $G_{RI}G_{IT}=10^{-16}$.

\begin{figure}[t]
\centering
\includegraphics[width=0.42\textwidth]{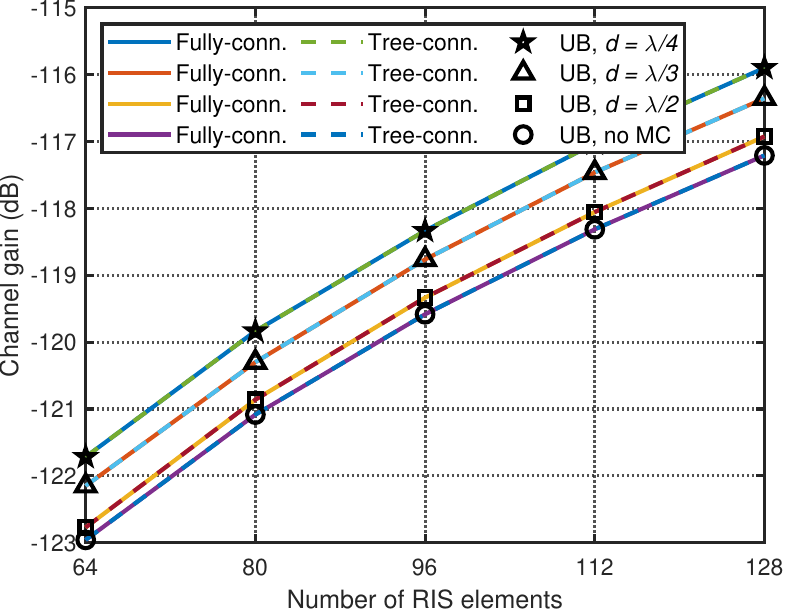}
\caption{Channel gain versus the number of RIS elements for different values of inter-element distance $d$.}
\label{fig:gain}
\end{figure}

\begin{figure*}[t]
\centering
\subfigure[Mean of $\vert\mathbf{z}_{RI}\Re\{\mathbf{Z}_{II}\}^{-1}\mathbf{z}_{IT}\vert^2$ and $\frac{1}{Z_{II}^2}\vert\mathbf{z}_{RI}\mathbf{z}_{IT}\vert^2$.]{
\includegraphics[width=0.42\textwidth]{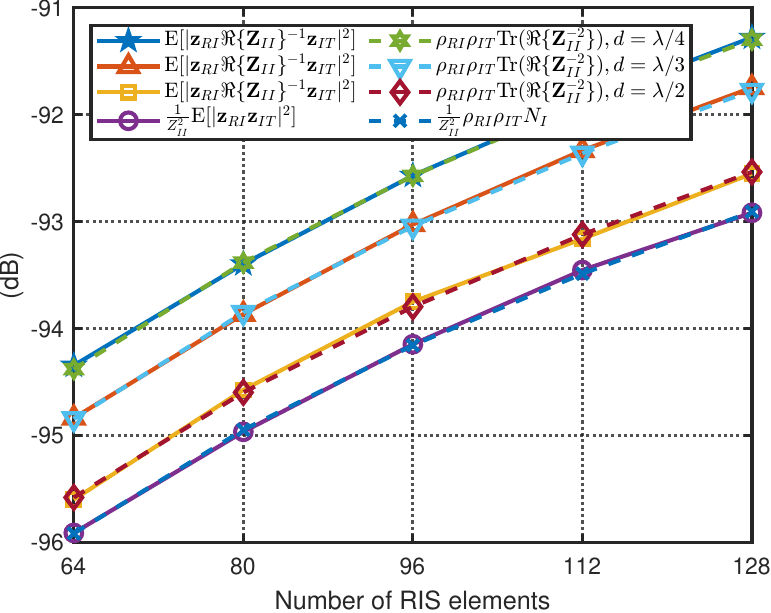}
}
\subfigure[Mean of $\Vert\mathbf{z}_{RI}\Re\{\mathbf{Z}_{II}\}^{-1/2}\Vert_2^2$ and $\frac{1}{Z_{II}}\Vert\mathbf{z}_{RI}\Vert_2^2$.]{
\includegraphics[width=0.42\textwidth]{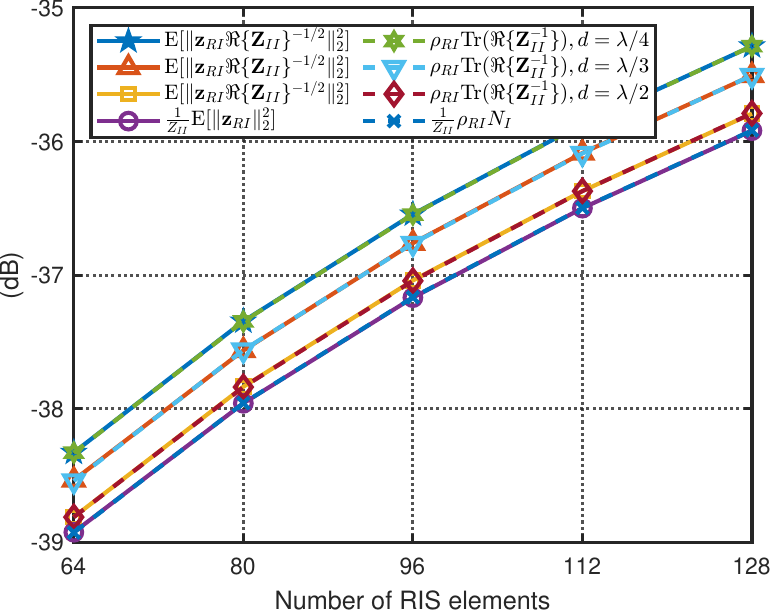}
}
\subfigure[Mean of $\vert\mathbf{z}_{RI}\Re\{\mathbf{Z}_{II}\}^{-1}\mathbf{z}_{IT}\vert$ and $\frac{1}{Z_{II}}\vert\mathbf{z}_{RI}\mathbf{z}_{IT}\vert$.]{
\includegraphics[width=0.42\textwidth]{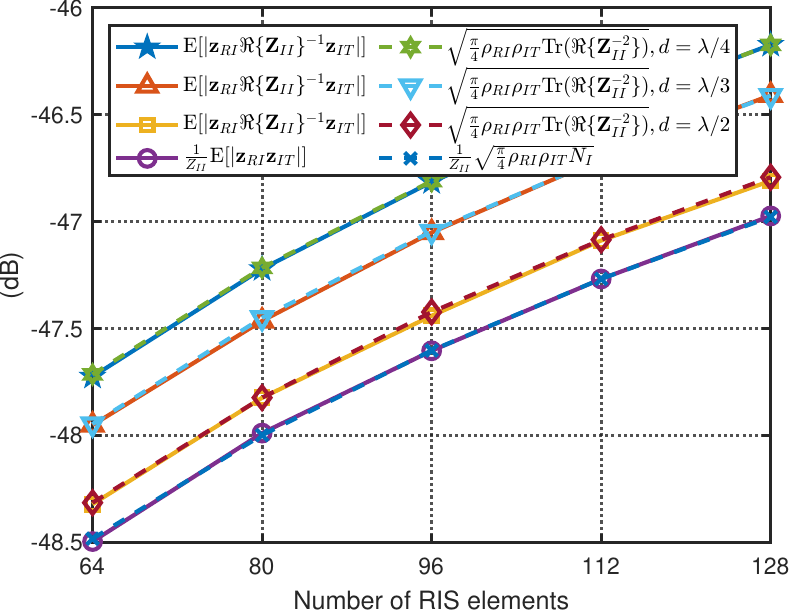}
}
\subfigure[Mean of $\Vert\mathbf{z}_{RI}\Re\{\mathbf{Z}_{II}\}^{-1/2}\Vert_2$ and $\frac{1}{\sqrt{Z_{II}}}\Vert\mathbf{z}_{RI}\Vert_2$.]{
\includegraphics[width=0.42\textwidth]{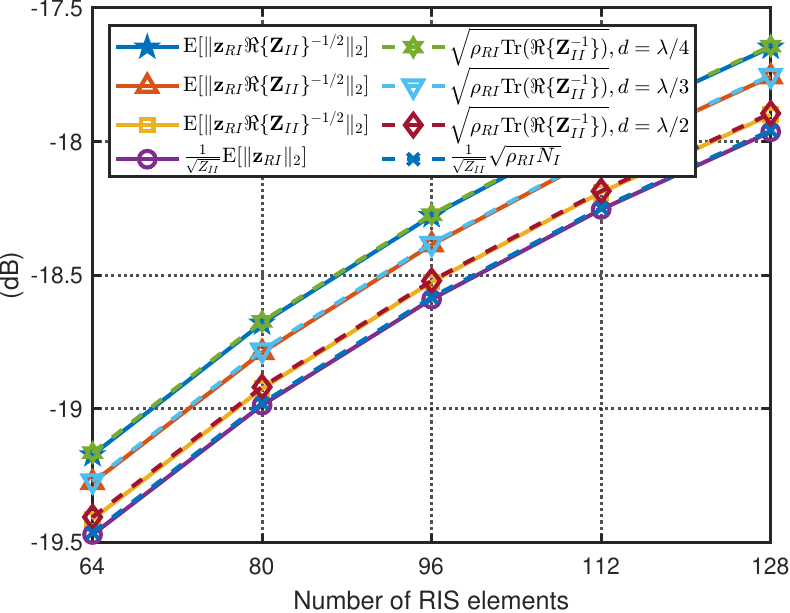}
}
\caption{Simulated mean values of the terms appearing in the channel gain expression and corresponding theoretical closed-form.}
\label{fig:scaling-laws-terms}
\end{figure*}

\begin{figure}[t]
\centering
\includegraphics[width=0.42\textwidth]{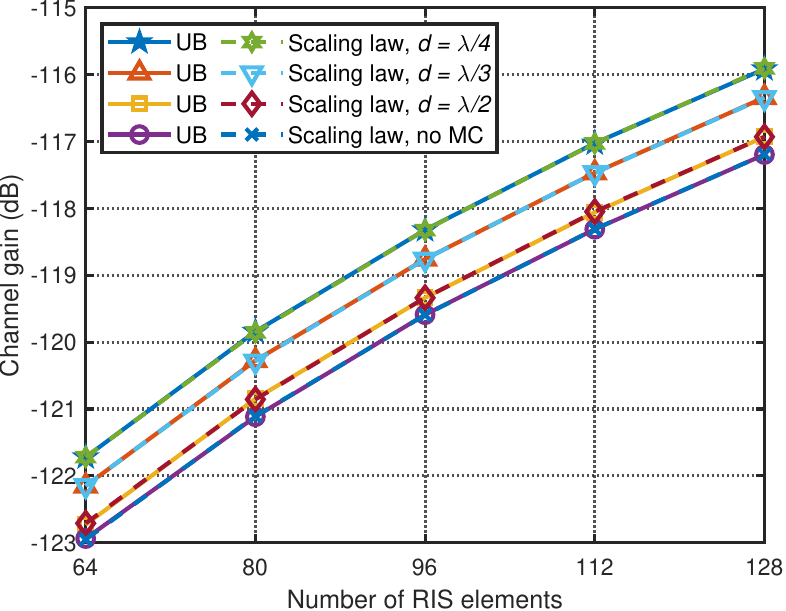}
\caption{Channel gain upper bounds and their scaling laws versus the number of RIS elements for different values of inter-element distance $d$.}
\label{fig:scaling-laws}
\end{figure}

\begin{figure*}[t]
\centering
\subfigure[$N_I=64$ RIS elements.]{
\includegraphics[width=0.42\textwidth]{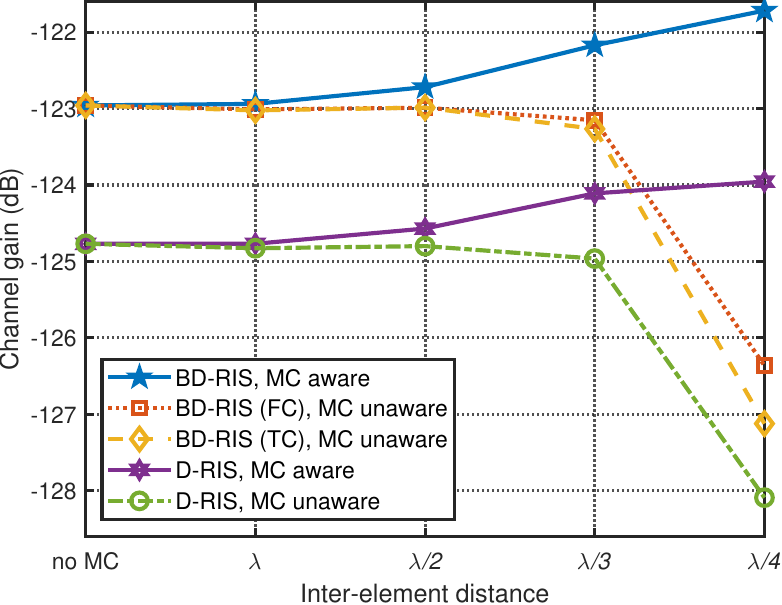}
\label{fig:MC-unaware-opt-Z-a}
}
\subfigure[$N_I=128$ RIS elements.]{
\includegraphics[width=0.42\textwidth]{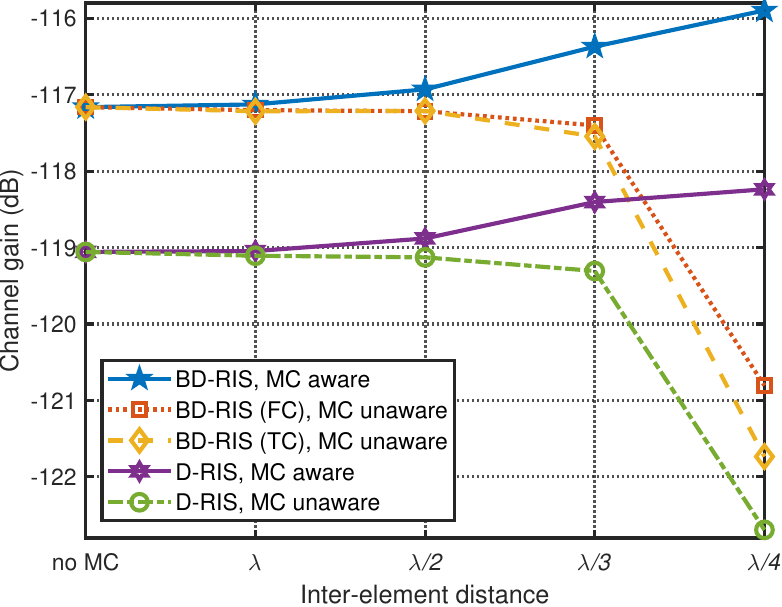}
\label{fig:MC-unaware-opt-Z-b}
}
\caption{Channel gain versus the inter-element distance $d$ achieved by BD-RIS and D-RIS reconfigured through mutual coupling-aware (denoted as ``MC aware'') and -unaware (denoted as ``MC unaware'') optimization.
In the case of coupling-unaware optimization, we assume the channels in the $Z$-parameters $\mathbf{z}_{RI}$, $\mathbf{z}_{IT}$, and $z_{RT}$ to be known.}
\label{fig:MC-unaware-opt-Z}
\end{figure*}

\subsection{Optimizing Fully- and Tree-Connected RISs}

% Figure
In Fig.~\ref{fig:gain}, we report the channel gain obtained by fully- and tree-connected RISs, for different values of inter-element distance $d\in[\lambda/2,\lambda/3,\lambda/4]$ together with its upper bound ``UB'' given by \eqref{eq:UB-MC}.
We also report the channel gain achieved by fully- and tree-connected RISs with no mutual coupling, i.e., with $\mathbf{Z}_{II}=Z_0\mathbf{I}$, by optimizing them with the global optimal solutions proposed in \cite{ner22} and \cite{ner23-1}, respectively, together with the channel gain upper bound ``UB'' given by \eqref{eq:UB-NoMC} where we set $Z_{II}=Z_0$.
%\begin{equation}
%\left\vert h^\star\right\vert^2
%=\frac{1}{4Z_0^2}
%\left(\left\vert z_{RT}-\frac{\mathbf{z}_{RI}\mathbf{z}_{IT}}{2Z_0}\right\vert
%+\frac{\left\Vert\mathbf{z}_{RI}\right\Vert_2\left\Vert\mathbf{z}_{IT}\right\Vert_2}{2Z_0}\right)^{2},
%\end{equation}
%obtained by setting $\mathbf{Z}_{II}=Z_0\mathbf{I}$ in \eqref{eq:UB-FC-3}.
We identify the following three remarks.
\textit{First}, both the fully-connected and the tree-connected RISs exactly achieve the channel gain upper bound, confirming the effectiveness of our global optimal closed-form solutions.
\textit{Second}, the presence of mutual coupling allows the BD-RISs to achieve higher channel gain over the case with no mutual coupling, as stated in Proposition~\ref{pro:impact}.
Thus, these numerical results confirm that mutual coupling is beneficial in BD-RIS-aided systems when it is exploited through mutual coupling-aware optimization algorithms.
\textit{Third}, smaller inter-element distances enable higher channel gains since they cause stronger mutual coupling.

\subsection{Scaling Laws}

In Section~\ref{sec:scaling-laws}, we have studied the mean value of the channel gain upper bounds \eqref{eq:UB-MC} and \eqref{eq:UB-NoMC}, i.e., their scaling laws, which are provided in closed-form in \eqref{eq:E-UB-MC-CF} and \eqref{eq:E-UB-NoMC-CF}, respectively.
We now verify the accuracy of these scaling laws.
To this end, in Fig.~\ref{fig:scaling-laws-terms}, we individually analyze the expectation terms appearing in \eqref{eq:E-UB-MC} and \eqref{eq:E-UB-NoMC}, reporting their simulated mean and the theoretical mean derived in closed-form in Section~\ref{sec:scaling-laws}.
We observe that the theoretical mean values are accurate for practical numbers of RIS elements, such as $N_I=64$, even if derived with approximations valid at $N_I\rightarrow\infty$.
Additional numerical results have shown that the scaling laws remain accurate for a number of RIS elements as low as $N_I=16$.

In Fig.~\ref{fig:scaling-laws}, we compare the simulated mean values of channel gain upper bounds ``UB'' given in \eqref{eq:UB-MC} and \eqref{eq:UB-NoMC} and their scaling laws provided by \eqref{eq:E-UB-MC-CF} and \eqref{eq:E-UB-NoMC-CF}, respectively.
We observe that the scaling laws provided in \eqref{eq:E-UB-MC-CF} and \eqref{eq:E-UB-NoMC-CF} exhibit remarkable accuracy when compared to the simulated mean values of the channel gain upper bounds.
This high level of agreement confirms that the scaling laws are not only theoretically sound but also highly reliable for practical scenarios, even with a finite number of RIS elements.
While Rayleigh distributed channels are considered in Fig.~\ref{fig:scaling-laws}, it can be shown that our scaling laws have the same level of accuracy also under Rician distributed channels, including Rayleigh and \gls{los} channels as specific cases.

\subsection{Benefits of BD-RIS and Mutual Coupling-Aware Optimization}

So far, we have analyzed the performance of BD-RIS-aided systems, where the BD-RIS is globally optimized accounting for the mutual coupling effects.
We now extend the comparison to include D-RIS-aided systems, as well as systems where the RIS is optimized without accounting for mutual coupling.
In Fig.~\ref{fig:MC-unaware-opt-Z}, we report the channel gain achieved by BD-RIS and D-RIS, when they are optimized through mutual coupling-aware as well as -unaware algorithms as detailed in the following.
\begin{itemize}
\item For BD-RIS reconfigured with mutual-coupling aware optimization, the performance is given by the channel gain upper bound in \eqref{eq:UB-MC}, both for fully- and tree-connected RISs.
\item For BD-RIS optimized in a mutual-coupling unaware fashion, we assume that $\mathbf{Z}_{II}=Z_0\mathbf{I}$ during the optimization phase, such that the channel \eqref{eq:h-Z} is given by
\begin{equation}
h=\frac{1}{2Z_0}(z_{RT}-\mathbf{z}_{RI}(j\mathbf{X}_{I}+Z_0\mathbf{I})^{-1}\mathbf{z}_{IT}).\label{eq:h-Z-unaware}
\end{equation}
Thus, following \cite{ner23-2}, we introduce
\begin{gather}
\mathbf{s}_{RI}=\frac{\mathbf{z}_{RI}}{2Z_0},\;
\mathbf{s}_{IT}=\frac{\mathbf{z}_{IT}}{2Z_0},\label{eq:sRI-sIT-unaware}\\
s_{RT}=\frac{1}{2Z_0}\left(z_{RT}-\frac{\mathbf{z}_{RI}\mathbf{z}_{IT}}{2Z_0}\right),\label{eq:sRT-unaware}\\
\boldsymbol{\Theta}=\left(j\mathbf{X}_I+Z_0\mathbf{I}\right)^{-1}\left(j\mathbf{X}_I-Z_0\mathbf{I}\right),\label{eq:T-unaware}
\end{gather}
to rewrite \eqref{eq:h-Z-unaware} equivalently as
\begin{equation}
h=s_{RT}+\mathbf{s}_{RI}\boldsymbol{\Theta}\mathbf{s}_{IT},\label{eq:h-S-unaware}
\end{equation}
and we optimize $\boldsymbol{\Theta}$ to maximize $\vert h\vert^2$ by using the global optimal solution proposed in \cite{ner22} and \cite{ner23-1} for fully- ``FC'' and tree-connected ``TC'' RISs, respectively.
From the obtained $\boldsymbol{\Theta}$, the susceptance matrix $\mathbf{X}_I$ is obtained by inverting \eqref{eq:T-unaware}, which is plugged into the channel model in \eqref{eq:h-Z} to get the channel gain of a BD-RIS-aided system with mutual coupling-unaware optimization.
\item For D-RIS optimized through a mutual-coupling aware method, we consider the optimization method proposed in \cite{qia21}, which has a computational complexity of $\mathcal{O}(IN_I^3)$, where $I$ is the number of iterations needed to converge.
\item For D-RIS optimized in a mutual-coupling unaware way, we rewrite the channel $h$ as in \eqref{eq:h-S-unaware} by introducing $\mathbf{s}_{RI}$, $\mathbf{s}_{IT}$, $s_{RT}$, and $\boldsymbol{\Theta}$ as in \eqref{eq:sRI-sIT-unaware}-\eqref{eq:T-unaware}.
Thus, $\boldsymbol{\Theta}=\mathrm{diag}(e^{j\theta_{1}},\ldots,e^{j\theta_{N_I}})$ is reconfigured by setting $\theta_{n_I}=\arg(s_{RT})-\arg([\mathbf{s}_{RI}]_{n_I}[\mathbf{s}_{IT}]_{n_I})$, $\forall n_I$, and $\mathbf{X}_I$ is accordingly computed by inverting \eqref{eq:T-unaware}.
\end{itemize}
We make the following observations from Fig.~\ref{fig:MC-unaware-opt-Z}.
\textit{First}, when optimizing a BD-RIS or D-RIS accounting for mutual coupling, the performance increases with the mutual coupling strength, i.e., as the inter-element distance decreases.
However, when BD-RIS or D-RIS are optimized not being aware of the mutual coupling, their performance dramatically drops for short inter-element distances, experiencing a degradation of up to 5~dB for BD-RIS and 4~dB for D-RIS.
This is because the RIS is optimized based on a channel model with no mutual coupling, which differs from the physics-consistent channel model including the mutual coupling effects.
\textit{Second}, the gain of BD-RIS over D-RIS is approximately 2~dB when optimized accounting for mutual coupling, which is $1.6$ times in linear scale, consistently with the gain derived in \cite{she20} in the absence of mutual coupling.
This gain slightly increases with the mutual coupling strength, i.e., as the inter-element distance decreases, since BD-RIS offers additional flexibility to deal with mutual coupling.
\textit{Third}, the observed trends of the channel gain versus the inter-element distance are independent of the number of RIS elements, as it can be seen by comparing Fig.~\ref{fig:MC-unaware-opt-Z-a} and \ref{fig:MC-unaware-opt-Z-b}, where we fix $N_I=64$ and $N_I=128$, respectively.
\textit{Fourth}, the impact of mutual coupling can be approximately neglected when the inter-element distance is larger than half-wavelength, i.e., $\lambda/2$, and the RIS can be successfully optimized without considering mutual coupling in this case.

%%%%%%%%%%%%%%%%%%%%%%%%%%%%%%%%%%%%%%%%%%%%%%%%%%
\section{Conclusion}
\label{sec:conclusion}

We globally optimize BD-RIS to maximize the channel gain in RIS-aided systems in the presence of \gls{em} mutual coupling between the RIS elements.
Specifically, we propose global optimal closed-form solutions to optimize fully- and tree-connected RISs.
We prove that fully- and tree-connected RISs achieve the same performance, confirming the benefit of the tree-connected RIS as a very low-complexity BD-RIS architecture that can achieve the maximum benefits.
In addition, we provide the expression of the channel gain achievable in the presence of mutual coupling and its scaling law in closed form.
Given this scaling law, we analytically prove that mutual coupling is beneficial in enhancing the average channel gain when considering Rayleigh distributed channels.

Numerical results are provided to support the theoretical derivations, showing that both fully- and tree-connected RISs with mutual coupling can achieve the channel gain upper bound.
Furthermore, we observe that stronger mutual coupling effects lead to a further enhanced channel gain, in agreement with the theoretical intuition.
We observe that BD-RIS offers a 2~dB gain over D-RIS in terms of channel gain and that mutual coupling-unaware optimization of the RIS can cause a channel gain degradation as high as 5~dB.

Future research directions include the extension of our results to multi-antenna and multi-user systems, which is challenging since closed-form solutions and scaling laws are often not known even in the absence of mutual coupling.
In addition, more realistic studies could analyze RIS in the presence of mutual coupling jointly with other hardware impairments and non-idealities, such as losses, discrete-value tunable components, and imperfect \gls{csi}.
Experimental validation is needed to verify whether the derived performance upper bound can be achieved in real-world scenarios.

\section*{Appendix}

\subsection{Proof of Proposition~\ref{pro:semi-definite}}

Since the network is lossy, the real power delivered to the network (dissipated) must be non-negative.
The real power dissipated by an $N$-port network is defined as
\begin{equation}
P=\frac{1}{2}\Re\left\{\mathbf{v}^T\mathbf{i}^*\right\},\label{eq:P}
\end{equation}
where $\mathbf{v}=[v_1,\ldots,v_N]^T\in\mathbb{C}^{N\times1}$ and $\mathbf{i}=[i_1,\ldots,i_N]^T\in\mathbb{C}^{N\times1}$ are the voltage and current vectors at the $N$ ports, respectively \cite[Chapter 4]{poz11}.
Furthermore, by recalling that $\mathbf{v}=\mathbf{Z}\mathbf{i}$, we have
\begin{align}
P
&=\frac{1}{2}\Re\left\{\mathbf{i}^T\mathbf{Z}\mathbf{i}^*\right\}\\
&=\frac{1}{2}\sum_{n=1}^N\Re\left\{\left\vert i_n\right\vert^2\left[\mathbf{Z}\right]_{n,n}+\sum_{m\neq n}i_ni_m^*\left[\mathbf{Z}\right]_{n,m}\right\},
\end{align}
where we exploited the fact that $\mathbf{Z}$ is symmetric for reciprocal networks.
Since $\mathbf{Z}$ is symmetric, we have $[\mathbf{Z}]_{m,n}=[\mathbf{Z}]_{n,m}$, yielding
\begin{equation}
P=\frac{1}{2}\sum_{n=1}^N\Re\left\{\left\vert i_n\right\vert^2\left[\mathbf{Z}\right]_{n,n}+\sum_{m>n}\left(i_ni_m^*+i_mi_n^*\right)\left[\mathbf{Z}\right]_{n,m}\right\},
\end{equation}
and noticing that $\vert i_n\vert^2$ and $i_ni_m^*+i_mi_n^*$ are purely real, we can write
\begin{multline}
P=\frac{1}{2}\sum_{n=1}^N\left(\left\vert i_n\right\vert^2\left[\Re\left\{\mathbf{Z}\right\}\right]_{n,n}\right.\\
\left.+\sum_{m>n}\left(i_ni_m^*+i_mi_n^*\right)\left[\Re\left\{\mathbf{Z}\right\}\right]_{n,m}\right).
\end{multline}
Thus, we can exploit again the symmetry of $\mathbf{Z}$ to express the real power $P$ as
\begin{equation}
P=\frac{1}{2}\mathbf{i}^T\Re\left\{\mathbf{Z}\right\}\mathbf{i}^*\label{eq:psp}.
\end{equation}
For a lossy network, the real power must be non-negative, i.e., $P\geq0$, for any current vector $\mathbf{i}$, indicating that $\Re\{\mathbf{Z}\}$ is positive semi-definite because of \eqref{eq:psp}.

A similar discussion can be also repeated to show that $\Re\{\mathbf{Y}\}$ is positive semi-definite.
To this end, we can exploit the relationship $\mathbf{i}=\mathbf{Y}\mathbf{v}$ to express the real power defined in \eqref{eq:P} as $P=\Re\{\mathbf{v}^T\mathbf{Y}^*\mathbf{v}^*\}/2$.
The rest of the proof is identical to the one provided for $\Re\{\mathbf{Z}\}$.

\subsection{Proof of Lemma~\ref{lem:1}}

% https://math.stackexchange.com/questions/3279641/minimize-the-sum-of-reciprocal-of-probabilities
To prove Lemma~\ref{lem:1}, we denote the $N$ eigenvalues of $\mathbf{A}$ as $\lambda_1,\ldots,\lambda_N$, where $\lambda_n>0$, for $n=1,\ldots,N$, given the positive definiteness of $\mathbf{A}$.
The sum of these eigenvalues is given by
\begin{equation}
\sum_{n=1}^N\lambda_n=\text{Tr}\left(\mathbf{A}\right),\label{eq:TA}
\end{equation}
which is $\text{Tr}\left(\mathbf{A}\right)=aN$ since $[\mathbf{A}]_{n,n}=a$, for $n=1,\ldots,N$.
Besides, $\lambda_1^{-1},\ldots,\lambda_N^{-1}$ are the $N$ eigenvalues of $\mathbf{A}^{-1}$, whose sum is given by
\begin{equation}
\sum_{n=1}^N\lambda_n^{-1}=\text{Tr}\left(\mathbf{A}^{-1}\right).\label{eq:TA-1}
\end{equation}
Following \eqref{eq:TA} and \eqref{eq:TA-1}, Lemma~\ref{lem:1} can be proved by showing that the minimization problem
\begin{equation}
\underset{\lambda_1,\ldots,\lambda_N}{\mathsf{\mathrm{min}}}\;\;
\sum_{n=1}^N\lambda_n^{-1}\;\;
\mathsf{\mathrm{s.t.}}\;\;
\sum_{n=1}^N\lambda_n=aN,\;\lambda_n>0,\label{eq:minTA-1}
\end{equation}
is solved when all the eigenvalues are equal, i.e., $\lambda_1=\ldots=\lambda_N=a$, giving $\sum_{n=1}^N\lambda_n^{-1}\geq N/a$.

Such a solution for the problem in \eqref{eq:minTA-1} can be proved by contradiction as follows.
Suppose that $\sum_{n=1}^N\lambda_n^{-1}$ is minimized when the $i$th and $j$th eigenvalues are $\lambda_i=\lambda_i^\star$ and $\lambda_j=\lambda_j^\star$, with $\lambda_i^\star\neq\lambda_j^\star$.
In this case, we have
\begin{equation}
\frac{1}{\lambda_i^\star}+
\frac{1}{\lambda_j^\star}>
\frac{1}{\frac{\lambda_i^\star+\lambda_j^\star}{2}}+
\frac{1}{\frac{\lambda_i^\star+\lambda_j^\star}{2}},
\end{equation}
following $(\lambda_i^\star-\lambda_j^\star)^2>0$.
Thus, the solution $\lambda_i=\lambda_j=(\lambda_i^\star+\lambda_j^\star)/2$ attains a lower value of $\sum_{n=1}^N\lambda_n^{-1}$, in contradiction with the hypothesis $\lambda_i=\lambda_i^\star$ and $\lambda_j=\lambda_j^\star$, with $\lambda_i^\star\neq\lambda_j^\star$.
This contradiction shows that \eqref{eq:minTA-1} is solved when $\lambda_1=\ldots=\lambda_N=a$, proving Lemma~\ref{lem:1}.

\subsection{Proof of Lemma~\ref{lem:2}}

%https://en.wikipedia.org/wiki/Trace_(linear_algebra)
This lemma follows from the Cauchy–Schwarz inequality applied to the trace operator and Lemma~\ref{lem:1}.
First, following the Cauchy–Schwarz inequality, we have
\begin{equation}
\text{Tr}\left(\mathbf{A}^{-2}\right)\geq\frac{\text{Tr}\left(\mathbf{A}^{-1}\right)^2}{N}.\label{eq:cs}
\end{equation}
Second, by substituting \eqref{eq:lem1} into \eqref{eq:cs}, we readily obtain \eqref{eq:lem2}.

\subsection{Proof of Proposition~\ref{pro:impact}}

To prove that $\mathbb{E}[\vert h_{MC}^\star\vert^2]\geq\mathbb{E}[\vert h_{NoMC}^\star\vert^2]$, it is sufficient to individually prove the following four inequalities
\begin{align}
\text{Tr}\left(\Re\{\mathbf{Z}_{II}\}^{-2}\right)&\geq\frac{N_I}{Z_{II}^2},\label{eq:1}\\
\text{Tr}\left(\Re\{\mathbf{Z}_{II}\}^{-1}\right)^2&\geq\frac{N_I^2}{Z_{II}^2},\label{eq:2}\\
\sqrt{\pi\text{Tr}\left(\Re\{\mathbf{Z}_{II}\}^{-2}\right)}&\geq\frac{\sqrt{\pi N_I}}{Z_{II}},\label{eq:3}\\
\text{Tr}\left(\Re\{\mathbf{Z}_{II}\}^{-1}\right)&\geq\frac{N}{Z_{II}},\label{eq:4}
\end{align}
which can be done by applying Lemmas~\ref{lem:1} and \ref{lem:2}.
The inequality in \eqref{eq:1} can be proved by applying Lemma~\ref{lem:2} with $\mathbf{A}=\Re\{\mathbf{Z}_{II}\}$ and $a=Z_{II}$.
Note that $\Re\{\mathbf{Z}_{II}\}$ fulfills the hypothesis of Lemma~\ref{lem:2} since it is positive definite following Proposition~\ref{pro:semi-definite}.
Besides, the inequality in \eqref{eq:2} can be proved by proving $\text{Tr}(\Re\{\mathbf{Z}_{II}\}^{-1})\geq N/Z_{II}$, which follows from Lemma~\ref{lem:1}.
Finally, the inequalities in \eqref{eq:3} and \eqref{eq:4} directly follows from \eqref{eq:1} and \eqref{eq:2}, respectively, concluding the proof that $\mathbb{E}[\vert h_{MC}^\star\vert^2]\geq\mathbb{E}[\vert h_{NoMC}^\star\vert^2]$.

\bibliographystyle{IEEEtran}
\bibliography{IEEEabrv,main}

% Generated by IEEEtran.bst, version: 1.14 (2015/08/26)
\begin{thebibliography}{10}
\providecommand{\url}[1]{#1}
\csname url@samestyle\endcsname
\providecommand{\newblock}{\relax}
\providecommand{\bibinfo}[2]{#2}
\providecommand{\BIBentrySTDinterwordspacing}{\spaceskip=0pt\relax}
\providecommand{\BIBentryALTinterwordstretchfactor}{4}
\providecommand{\BIBentryALTinterwordspacing}{\spaceskip=\fontdimen2\font plus
\BIBentryALTinterwordstretchfactor\fontdimen3\font minus \fontdimen4\font\relax}
\providecommand{\BIBforeignlanguage}[2]{{%
\expandafter\ifx\csname l@#1\endcsname\relax
\typeout{** WARNING: IEEEtran.bst: No hyphenation pattern has been}%
\typeout{** loaded for the language `#1'. Using the pattern for}%
\typeout{** the default language instead.}%
\else
\language=\csname l@#1\endcsname
\fi
#2}}
\providecommand{\BIBdecl}{\relax}
\BIBdecl

\bibitem{wu21}
Q.~Wu, S.~Zhang, B.~Zheng, C.~You, and R.~Zhang, ``Intelligent reflecting surface-aided wireless communications: A tutorial,'' \emph{IEEE Trans. Commun.}, vol.~69, no.~5, pp. 3313--3351, 2021.

\bibitem{li23-1}
H.~Li, S.~Shen, M.~Nerini, and B.~Clerckx, ``Reconfigurable intelligent surfaces 2.0: Beyond diagonal phase shift matrices,'' \emph{IEEE Commun. Mag.}, vol.~62, no.~3, pp. 102--108, 2024.

\bibitem{she20}
S.~Shen, B.~Clerckx, and R.~Murch, ``Modeling and architecture design of reconfigurable intelligent surfaces using scattering parameter network analysis,'' \emph{IEEE Trans. Wireless Commun.}, vol.~21, no.~2, pp. 1229--1243, 2022.

\bibitem{ner23-1}
M.~Nerini, S.~Shen, H.~Li, and B.~Clerckx, ``Beyond diagonal reconfigurable intelligent surfaces utilizing graph theory: Modeling, architecture design, and optimization,'' \emph{IEEE Trans. Wireless Commun.}, 2024.

\bibitem{ner23-3}
M.~Nerini and B.~Clerckx, ``Pareto frontier for the performance-complexity trade-off in beyond diagonal reconfigurable intelligent surfaces,'' \emph{IEEE Commun. Lett.}, vol.~27, no.~10, pp. 2842--2846, 2023.

\bibitem{li23-2}
H.~Li, S.~Shen, and B.~Clerckx, ``Beyond diagonal reconfigurable intelligent surfaces: From transmitting and reflecting modes to single-, group-, and fully-connected architectures,'' \emph{IEEE Trans. Wireless Commun.}, vol.~22, no.~4, pp. 2311--2324, 2023.

\bibitem{li23-3}
------, ``Beyond diagonal reconfigurable intelligent surfaces: A multi-sector mode enabling highly directional full-space wireless coverage,'' \emph{IEEE J. Sel. Areas Commun.}, vol.~41, no.~8, pp. 2446--2460, 2023.

\bibitem{li22}
Q.~Li \emph{et~al.}, ``Reconfigurable intelligent surfaces relying on non-diagonal phase shift matrices,'' \emph{IEEE Trans. Veh. Technol.}, vol.~71, no.~6, pp. 6367--6383, 2022.

\bibitem{wan24}
H.~Wang, Z.~Han, and A.~L. Swindlehurst, ``Channel reciprocity attacks using intelligent surfaces with non-diagonal phase shifts,'' \emph{IEEE Open J. Commun. Soc.}, vol.~5, pp. 1469--1485, 2024.

\bibitem{ner22}
M.~Nerini, S.~Shen, and B.~Clerckx, ``Closed-form global optimization of beyond diagonal reconfigurable intelligent surfaces,'' \emph{IEEE Trans. Wireless Commun.}, vol.~23, no.~2, pp. 1037--1051, 2024.

\bibitem{fan24}
T.~Fang and Y.~Mao, ``A low-complexity beamforming design for beyond-diagonal {RIS} aided multi-user networks,'' \emph{IEEE Commun. Lett.}, vol.~28, no.~1, pp. 203--207, 2024.

\bibitem{gra21}
G.~Gradoni and M.~Di~Renzo, ``End-to-end mutual coupling aware communication model for reconfigurable intelligent surfaces: An electromagnetic-compliant approach based on mutual impedances,'' \emph{IEEE Wireless Commun. Lett.}, vol.~10, no.~5, pp. 938--942, 2021.

\bibitem{dir23}
M.~Di~Renzo, V.~Galdi, and G.~Castaldi, ``Modeling the mutual coupling of reconfigurable metasurfaces,'' in \emph{2023 17th European Conference on Antennas and Propagation (EuCAP)}, 2023.

\bibitem{ner23-2}
M.~Nerini, S.~Shen, H.~Li, M.~Di~Renzo, and B.~Clerckx, ``A universal framework for multiport network analysis of reconfigurable intelligent surfaces,'' \emph{IEEE Trans. Wireless Commun.}, 2024.

\bibitem{qia21}
X.~Qian and M.~Di~Renzo, ``Mutual coupling and unit cell aware optimization for reconfigurable intelligent surfaces,'' \emph{IEEE Wireless Commun. Lett.}, vol.~10, no.~6, pp. 1183--1187, 2021.

\bibitem{per23}
N.~S. Perović, L.-N. Tran, M.~Di~Renzo, and M.~F. Flanagan, ``Optimization of {RIS}-aided {SISO} systems based on a mutually coupled loaded wire dipole model,'' in \emph{2023 57th Asilomar Conference on Signals, Systems, and Computers}, 2023, pp. 145--150.

\bibitem{abr21}
A.~Abrardo, D.~Dardari, M.~Di~Renzo, and X.~Qian, ``{MIMO} interference channels assisted by reconfigurable intelligent surfaces: Mutual coupling aware sum-rate optimization based on a mutual impedance channel model,'' \emph{IEEE Wireless Commun. Lett.}, vol.~10, no.~12, pp. 2624--2628, 2021.

\bibitem{akr23}
M.~Akrout, F.~Bellili, A.~Mezghani, and J.~A. Nossek, ``Physically consistent models for intelligent reflective surface-assisted communications under mutual coupling and element size constraint,'' in \emph{2023 57th Asilomar Conference on Signals, Systems, and Computers}, 2023, pp. 1589--1594.

\bibitem{ma23}
R.~Ma, J.~Tang, X.~Zhang, K.-K. Wong, and J.~A. Chambers, ``Energy-efficiency optimization for mutual-coupling-aware wireless communication system based on {RIS}-enhanced {SWIPT},'' \emph{IEEE Internet Things J.}, vol.~10, no.~22, pp. 19\,399--19\,414, 2023.

\bibitem{wij24}
D.~Wijekoon, A.~Mezghani, G.~C. Alexandropoulos, and E.~Hossain, ``Electromagnetically-consistent modeling and optimization of mutual coupling in {RIS}-assisted multi-user {MIMO} communication systems,'' in \emph{2024 IEEE International Conference on Communications Workshops (ICC Workshops)}, 2024, pp. 1737--1742.

\bibitem{zhe24-1}
P.~Zheng, X.~Ma, and T.~Y. Al-Naffouri, ``On the impact of mutual coupling on {RIS}-assisted channel estimation,'' \emph{IEEE Wireless Commun. Lett.}, vol.~13, no.~5, pp. 1275--1279, 2024.

\bibitem{zhe24-3}
P.~Zheng, S.~Tarboush, H.~Sarieddeen, and T.~Y. Al-Naffouri, ``Mutual coupling-aware channel estimation and beamforming for {RIS}-assisted communications,'' \emph{arXiv preprint arXiv:2410.04110}, 2024.

\bibitem{pen24}
B.~Peng, K.-L. Besser, S.~Shen, F.~Siegismund-Poschmann, R.~Raghunath, D.~Mittleman, V.~Jamali, and E.~A. Jorswieck, ``{RISnet}: A domain-knowledge driven neural network architecture for {RIS} optimization with mutual coupling and partial {CSI},'' \emph{arXiv preprint arXiv:2403.04028}, 2024.

\bibitem{sem24}
D.~Semmler, J.~A. Nossek, M.~Joham, and W.~Utschick, ``Performance analysis of systems with coupled and decoupled {RISs},'' in \emph{2024 19th International Symposium on Wireless Communication Systems (ISWCS)}, 2024.

\bibitem{mur23}
P.~Mursia, S.~Phang, V.~Sciancalepore, G.~Gradoni, and M.~Di~Renzo, ``{SARIS}: Scattering aware reconfigurable intelligent surface model and optimization for complex propagation channels,'' \emph{IEEE Wireless Commun. Lett.}, vol.~12, no.~11, pp. 1921--1925, 2023.

\bibitem{has24}
H.~E. Hassani, X.~Qian, S.~Jeong, N.~S. Perović, M.~Di~Renzo, P.~Mursia, V.~Sciancalepore, and X.~Costa-Pérez, ``Optimization of {RIS}-aided {MIMO}—{A} mutually coupled loaded wire dipole model,'' \emph{IEEE Wireless Commun. Lett.}, vol.~13, no.~3, pp. 726--730, 2024.

\bibitem{pet23}
G.~Pettanice, R.~Valentini, P.~D. Marco, F.~Loreto, D.~Romano, F.~Santucci, D.~Spina, and G.~Antonini, ``Mutual coupling aware time-domain characterization and performance analysis of reconfigurable intelligent surfaces,'' \emph{IEEE Trans. Electromagn. Compat.}, vol.~65, no.~6, pp. 1606--1620, 2023.

\bibitem{zhe24-2}
P.~Zheng, R.~Wang, A.~Shamim, and T.~Y. Al-Naffouri, ``Mutual coupling in {RIS}-aided communication: Model training and experimental validation,'' \emph{IEEE Trans. Wireless Commun.}, 2024.

\bibitem{li24}
H.~Li, S.~Shen, M.~Nerini, M.~Di~Renzo, and B.~Clerckx, ``Beyond diagonal reconfigurable intelligent surfaces with mutual coupling: Modeling and optimization,'' \emph{IEEE Commun. Lett.}, vol.~28, no.~4, pp. 937--941, 2024.

\bibitem{poz11}
D.~M. Pozar, \emph{Microwave engineering}.\hskip 1em plus 0.5em minus 0.4em\relax John wiley \& sons, 2011.

\bibitem{ivr10}
M.~T. Ivrlač and J.~A. Nossek, ``Toward a circuit theory of communication,'' \emph{IEEE Trans. Circuits Syst. I: Regul. Pap.}, vol.~57, no.~7, pp. 1663--1683, 2010.

\bibitem{ngo17}
H.~Q. Ngo and E.~G. Larsson, ``No downlink pilots are needed in {TDD} massive {MIMO},'' \emph{IEEE Trans. Wireless Commun.}, vol.~16, no.~5, pp. 2921--2935, 2017.

\bibitem{ahl53}
L.~V. Ahlfors, \emph{Complex analysis: An introduction to the theory of analytic functions of one complex variable}.\hskip 1em plus 0.5em minus 0.4em\relax McGraw-Hill, 1953.

\end{thebibliography}

\end{document}